\documentclass[10pt]{article}   
\usepackage{authblk}                  
\usepackage{hyperref}

\usepackage{rotating}
\usepackage{floatpag}
\usepackage{xifthen}
\usepackage{bm}

\rotfloatpagestyle{empty}
\renewcommand{\theenumi}{\arabic{enumi}}
\ifthenelse{\isundefined{\showoptional}}{
  \newcommand{\showoptional}{1}
}{}

\ifthenelse{\isundefined{\ismain}}{
  \newcommand{\ismain}{0}
}{}

\ifthenelse{\isundefined{\lecturenotes}}{
  \newcommand{\lecturenotes}{0}
}{}

\usepackage{hyperref}

\usepackage{amsmath,amsthm,amssymb,mathabx}

\usepackage{xspace,enumerate,color,epsfig} 
\usepackage{graphicx}
\graphicspath{{.}{./figures/}}

\usepackage{tikzfig}
\usepackage{stmaryrd}
\usepackage{docmute}
\usepackage{keycommand}

\newcommand{\emptydiag}{\,\tikz{\node[style=empty diagram] (x) {};}\,}

\newkeycommand{\boxmap}[style=small box][1]{\,\tikz{\node[style=\commandkey{style}] (x) {$#1$};\draw(0,-1.25)--(x)--(0,1.25);}\,}
\newkeycommand{\smallboxmap}[style=small box][1]{\,\tikz{\node[style=\commandkey{style}] (x) {$#1$};\draw(0,-1.25)--(x)--(0,1.25);}\,}
\newkeycommand{\shortboxmap}[style=small box][1]{\,\tikz{\node[style=\commandkey{style}] (x) {$#1$};\draw(0,-1)--(x)--(0,1);}\,}
\newkeycommand{\dboxmap}[style=dbox][1]{\,\tikz{\node[style=\commandkey{style}] (x) {$#1$};\draw[boldedge] (0,-1.25)--(x)--(0,1.25);}\,}
\newkeycommand{\shortdboxmap}[style=dbox][1]{\,\tikz{\node[style=\commandkey{style}] (x) {$#1$};\draw[boldedge] (0,-1)--(x)--(0,1);}\,}

\newcommand{\dmap}[1]{\,\tikz{\node[style=dmap] (x) {$#1$};\draw[boldedge] (0,-1.3)--(x)--(0,1.3);}\,}

\newcommand{\map}[1]{\,\tikz{\node[style=map] (x) {$#1$};\draw(0,-1.3)--(x)--(0,1.3);}\,}

\newkeycommand{\maptypes}[style=map][3]{\begin{tikzpicture}
  \begin{pgfonlayer}{nodelayer}
    \node [style=\commandkey{style}] (0) at (0, 0) {$#1$};
    \node [style=none] (1) at (0, -1.5) {};
    \node [style=none] (2) at (0, 1.5) {};
    \node [style=right label] (3) at (0.25, -1.25) {$#2$};
    \node [style=right label] (4) at (0.25, 1) {$#3$};
  \end{pgfonlayer}
  \begin{pgfonlayer}{edgelayer}
    \draw (1.center) to (0);
    \draw (0) to (2.center);
  \end{pgfonlayer}
\end{tikzpicture}}

\newcommand{\mapconj}[1]{\,\tikz{\node[style=mapconj] (x) {$#1$};\draw(0,-1.3)--(x)--(0,1.3);}\,}

\newkeycommand{\wirelabel}[style=white label]{\,\tikz{\node[style=\commandkey{style}]{};}\,\xspace}

\newkeycommand{\pointmap}[style=point][1]{\,\tikz{\node[style=\commandkey{style}] (x) at (0,-0.3) {$#1$};\node [style=none] (3) at (0, 0.9) {};\node [style=none] (3) at (0, -0.9) {};\draw (x)--(0,0.7);}\,}

\newkeycommand{\point}[style=point,estyle=][1]{\,\tikz{\node[style=\commandkey{style}] (x) at (0,-0.05) {$#1$};\draw [\commandkey{estyle}] (x)--(0,1.05);}\,}

\newkeycommand{\copointmap}[style=copoint][1]{\,\tikz{\node[style=\commandkey{style}] (x) at (0,0.3) {$#1$};\draw (x)--(0,-0.7);}\,}

\newkeycommand{\dpoint}[style=point][1]{\,\tikz{\node[style=\commandkey{style},doubled] (x) at (0,-0.3) {$#1$};\draw[boldedge] (x)--(0,0.7);}\,}

\newkeycommand{\kpoint}[style=kpoint,estyle=][1]{\,\tikz{\node[style=\commandkey{style}] (x) at (0,-0.05) {$#1$};\draw [\commandkey{estyle}] (x)--(0,1.05);}\,}

\newkeycommand{\kpointgray}[style=gray kpoint,estyle=][1]{\,\tikz{\node[style=\commandkey{style}] (x) at (0,-0.05) {$#1$};\draw [\commandkey{estyle}] (x)--(0,1.05);}\,}

\newkeycommand{\typedkpoint}[style=kpoint,edgestyle=][2]{%
\,\begin{tikzpicture}
  \begin{pgfonlayer}{nodelayer}
    \node [style=none] (0) at (0, 1) {};
    \node [style=\commandkey{style}] (1) at (0, -0.75) {$#2$};
    \node [style=right label] (2) at (0.25, 0.5) {$#1$};
  \end{pgfonlayer}
  \begin{pgfonlayer}{edgelayer}
    \draw [\commandkey{edgestyle}] (1) to (0.center);
  \end{pgfonlayer}
\end{tikzpicture}\,}

\newkeycommand{\typedkpointdag}[style=kpoint adjoint,edgestyle=][2]{%
\,\begin{tikzpicture}
  \begin{pgfonlayer}{nodelayer}
    \node [style=none] (0) at (0, -1) {};
    \node [style=\commandkey{style}] (1) at (0, 0.75) {$#2$};
    \node [style=right label] (2) at (0.25, -0.5) {$#1$};
  \end{pgfonlayer}
  \begin{pgfonlayer}{edgelayer}
    \draw [\commandkey{edgestyle}] (0.center) to (1);
  \end{pgfonlayer}
\end{tikzpicture}\,}

\newkeycommand{\bistate}[style=kpoint][1]{\,\begin{tikzpicture}
  \begin{pgfonlayer}{nodelayer}
    \node [style=\commandkey{style}, minimum width=1 cm, inner sep=2pt] (0) at (0, -0.25) {$#1$};
    \node [style=none] (1) at (-0.75, 0) {};
    \node [style=none] (2) at (0.75, 0) {};
    \node [style=none] (3) at (-0.75, 0.88) {};
    \node [style=none] (4) at (0.75, 0.88) {};
  \end{pgfonlayer}
  \begin{pgfonlayer}{edgelayer}
    \draw (1.center) to (3.center);
    \draw (2.center) to (4.center);
  \end{pgfonlayer}
\end{tikzpicture}\,}

\newkeycommand{\bistatebig}[style=kpoint][1]{\,\begin{tikzpicture}
  \begin{pgfonlayer}{nodelayer}
    \node [style=\commandkey{style}, minimum width=, minimum width=1.26 cm, inner sep=2pt] (0) at (0, -0.25) {$#1$};
    \node [style=none] (1) at (-1, 0) {};
    \node [style=none] (2) at (1, 0) {};
    \node [style=none] (3) at (-1, 0.88) {};
    \node [style=none] (4) at (1, 0.88) {};
  \end{pgfonlayer}
  \begin{pgfonlayer}{edgelayer}
    \draw (1.center) to (3.center);
    \draw (2.center) to (4.center);
  \end{pgfonlayer}
\end{tikzpicture}\,}

\newkeycommand{\dbistate}[style=dkpoint][1]{\,\begin{tikzpicture}
  \begin{pgfonlayer}{nodelayer}
    \node [style=\commandkey{style}, minimum width=1 cm, inner sep=2pt] (0) at (0, -0.25) {$#1$};
    \node [style=none] (1) at (-0.75, 0) {};
    \node [style=none] (2) at (0.75, 0) {};
    \node [style=none] (3) at (-0.75, 1.0) {};
    \node [style=none] (4) at (0.75, 1.0) {};
  \end{pgfonlayer}
  \begin{pgfonlayer}{edgelayer}
    \draw [boldedge] (1.center) to (3.center);
    \draw [boldedge] (2.center) to (4.center);
  \end{pgfonlayer}
\end{tikzpicture}\,}

\newkeycommand{\bistatebraket}[style1=kpoint,style2=kpointadj][2]{\,\begin{tikzpicture}
  \begin{pgfonlayer}{nodelayer}
    \node [style=\commandkey{style1}, minimum width=1 cm, inner sep=2pt] (0) at (0, -0.75) {$#1$};
    \node [style=none] (1) at (-0.75, -0.5) {};
    \node [style=none] (2) at (0.75, -0.5) {};
    \node [style=none] (3) at (-0.75, 0.5) {};
    \node [style=none] (4) at (0.75, 0.5) {};
    \node [style=\commandkey{style2}, minimum width=1 cm, inner sep=2pt] (5) at (0, 0.75) {$#2$};
  \end{pgfonlayer}
  \begin{pgfonlayer}{edgelayer}
    \draw (1.center) to (3.center);
    \draw (2.center) to (4.center);
  \end{pgfonlayer}
\end{tikzpicture}\,}

\newkeycommand{\weight}[style=dkpoint][1]{\,\begin{tikzpicture}
  \begin{pgfonlayer}{nodelayer}
    \node [style=\commandkey{style}] (0) at (0, -0.5) {$#1$};
    \node [style=upground] (1) at (0, 0.75) {};
  \end{pgfonlayer}
  \begin{pgfonlayer}{edgelayer}
    \draw [style=boldedge] (0) to (1);
  \end{pgfonlayer}
\end{tikzpicture}\,}

\newkeycommand{\dkpoint}[style=dkpoint][1]{\,\tikz{\node[style=\commandkey{style}] (x) at (0,-0.1) {$#1$};\draw[boldedge] (x)--(0,1);}\,}
\newkeycommand{\dkpointadj}[style=dkpointadj][1]{\,\tikz{\node[style=\commandkey{style}] (x) at (0,0.1) {$#1$};\draw[boldedge] (x)--(0,-1);}\,}
\newkeycommand{\dkpointtrans}[style=dkpointtrans][1]{\,\tikz{\node[style=\commandkey{style}] (x) at (0,0.1) {$#1$};\draw[boldedge] (x)--(0,-1);}\,}
\newkeycommand{\dkpointconj}[style=dkpointconj][1]{\,\tikz{\node[style=\commandkey{style}] (x) at (0,-0.1) {$#1$};\draw[boldedge] (x)--(0,1);}\,}

\newkeycommand{\blackkpoint}[style=black kpoint][1]{\,\tikz{\node[style=\commandkey{style}] (x) at (0,-0.1) {$#1$};\draw (x)--(0,1);}\,}
\newkeycommand{\blackkpointadj}[style=black kpointadj][1]{\,\tikz{\node[style=\commandkey{style}] (x) at (0,0.1) {$#1$};\draw (x)--(0,-1);}\,}
\newkeycommand{\blackdkpoint}[style=black dkpoint][1]{\,\tikz{\node[style=\commandkey{style}] (x) at (0,-0.1) {$#1$};\draw[boldedge] (x)--(0,1);}\,}
\newkeycommand{\blackdkpointadj}[style=black dkpointadj][1]{\,\tikz{\node[style=\commandkey{style}] (x) at (0,0.1) {$#1$};\draw[boldedge] (x)--(0,-1);}\,}

\newcommand{\trace}{\,\begin{tikzpicture}
  \begin{pgfonlayer}{nodelayer}
    \node [style=none] (0) at (0, -0.60) {};
    \node [style=upground] (1) at (0, 0.40) {};
  \end{pgfonlayer}
  \begin{pgfonlayer}{edgelayer}
    \draw [style=boldedge] (0.center) to (1);
  \end{pgfonlayer}
\end{tikzpicture}\,}

\newcommand\discard\trace

\newcommand\sqrtD{\textrm{\footnotesize $\sqrt{D}$}\xspace}

\newcommand\oneoverD{\ensuremath{{\textstyle{1\over{D}}}}\xspace}

\newcommand\oneoversqrtD{\ensuremath{{\textstyle{1\over{\sqrt{D}}}}}\xspace}

\newcommand{\namedeq}[1]{\,\begin{tikzpicture}
  \begin{pgfonlayer}{nodelayer}
    \node [style=none] (0) at (0, 0.75) {\footnotesize $#1$};
    \node [style=none] (1) at (0, 0) {$=$};
  \end{pgfonlayer}
\end{tikzpicture}\,}

\newcommand{\namedscalareq}[1]{\,\begin{tikzpicture}
  \begin{pgfonlayer}{nodelayer}
    \node [style=none] (0) at (0, 0.75) {\footnotesize $#1$};
    \node [style=none] (1) at (0, 0) {$\scalareq$};
  \end{pgfonlayer}
\end{tikzpicture}\,}

\newcommand{\scalareq}{\ensuremath{\approx}\xspace}

\newkeycommand{\pointketbra}[style1=copoint,style2=point,estyle=][2]{\,%
\begin{tikzpicture}
  \begin{pgfonlayer}{nodelayer}
    \node [style=none] (0) at (0, -1.75) {};
    \node [style=\commandkey{style1}] (1) at (0, -0.75) {$#1$};
    \node [style=\commandkey{style2}] (2) at (0, 0.75) {$#2$};
    \node [style=none] (3) at (0, 1.75) {};
  \end{pgfonlayer}
  \begin{pgfonlayer}{edgelayer}
    \draw [\commandkey{estyle}] (0.center) to (1);
    \draw [\commandkey{estyle}] (2) to (3.center);
  \end{pgfonlayer}
\end{tikzpicture}\,}

\newkeycommand{\twocopointketbra}[style1=copoint,style2=copoint,style3=point][3]{\,%
\begin{tikzpicture}
  \begin{pgfonlayer}{nodelayer}
    \node [style=none] (0) at (-0.75, -1.75) {};
    \node [style=none] (1) at (0.75, -1.75) {};
    \node [style=\commandkey{style1}] (2) at (-0.75, -0.75) {$#1$};
    \node [style=\commandkey{style2}] (3) at (0.75, -0.75) {$#2$};
    \node [style=\commandkey{style3}] (4) at (0, 0.75) {$#3$};
    \node [style=none] (5) at (0, 1.75) {};
  \end{pgfonlayer}
  \begin{pgfonlayer}{edgelayer}
    \draw (0.center) to (2);
    \draw (1.center) to (3);
    \draw (4) to (5.center);
  \end{pgfonlayer}
\end{tikzpicture}\,}

\newkeycommand{\twopointketbra}[style1=copoint,style2=point,style3=point][3]{\,%
\begin{tikzpicture}
  \begin{pgfonlayer}{nodelayer}
    \node [style=none] (0) at (0, -1.75) {};
    \node [style=none] (1) at (-0.75, 1.75) {};
    \node [style=\commandkey{style1}] (2) at (0, -0.75) {$#1$};
    \node [style=\commandkey{style2}] (3) at (-0.75, 0.75) {$#2$};
    \node [style=\commandkey{style3}] (4) at (0.75, 0.75) {$#3$};
    \node [style=none] (5) at (0.75, 1.75) {};
  \end{pgfonlayer}
  \begin{pgfonlayer}{edgelayer}
    \draw (0.center) to (2);
    \draw (1.center) to (3);
    \draw (4) to (5.center);
  \end{pgfonlayer}
\end{tikzpicture}\,}

\newcommand{\idwire}{\,%
\begin{tikzpicture}
  \begin{pgfonlayer}{nodelayer}
    \node [style=none] (0) at (0, -1.5) {};
    \node [style=none] (3) at (0, 1.5) {};
  \end{pgfonlayer}
  \begin{pgfonlayer}{edgelayer}
    \draw (0.center) to (3.center);
  \end{pgfonlayer}
\end{tikzpicture}\,}

\newcommand{\didwire}{\,%
\begin{tikzpicture}
  \begin{pgfonlayer}{nodelayer}
    \node [style=none] (0) at (0, -1.5) {};
    \node [style=none] (3) at (0, 1.5) {};
  \end{pgfonlayer}
  \begin{pgfonlayer}{edgelayer}
    \draw [boldedge] (0.center) to (3.center);
  \end{pgfonlayer}
\end{tikzpicture}\,}

\newkeycommand{\pointbraket}[style1=point,style2=copoint,estyle=][2]{\,%
\begin{tikzpicture}
  \begin{pgfonlayer}{nodelayer}
    \node [style=\commandkey{style1}] (0) at (0, -0.625) {$#1$};
    \node [style=\commandkey{style2}] (1) at (0, 0.625) {$#2$};
  \end{pgfonlayer}
  \begin{pgfonlayer}{edgelayer}
    \draw [\commandkey{estyle}] (0) to (1);
  \end{pgfonlayer}
\end{tikzpicture}\,}

\newkeycommand{\kpointketbra}[style1=kpointdag,style2=kpoint,estyle=][2]{\,%
\begin{tikzpicture}
  \begin{pgfonlayer}{nodelayer}
    \node [style=none] (0) at (0, -1.9) {};
    \node [style=\commandkey{style1}] (1) at (0, -0.95) {$#1$};
    \node [style=\commandkey{style2}] (2) at (0, 0.95) {$#2$};
    \node [style=none] (3) at (0, 1.9) {};
  \end{pgfonlayer}
  \begin{pgfonlayer}{edgelayer}
    \draw [\commandkey{estyle}] (0.center) to (1);
    \draw [\commandkey{estyle}] (2) to (3.center);
  \end{pgfonlayer}
\end{tikzpicture}\,}

\newkeycommand{\kpointmap}[style1=kpoint,style2=map][2]{\,%
\begin{tikzpicture}
  \begin{pgfonlayer}{nodelayer}
    \node [style=\commandkey{style1}] (0) at (0, -1.1) {$#1$};
    \node [style=\commandkey{style2}] (1) at (0, 0.5) {$#2$};
    \node [style=none] (2) at (0, 1.75) {};
  \end{pgfonlayer}
  \begin{pgfonlayer}{edgelayer}
    \draw (0) to (1);
    \draw (1) to (2.center);
  \end{pgfonlayer}
\end{tikzpicture}%
\,}

\newkeycommand{\kpointmaptrans}[style1=kpointconj,style2=mapadj][2]{\,%
\begin{tikzpicture}
  \begin{pgfonlayer}{nodelayer}
    \node [style=\commandkey{style1}] (0) at (0, -1.1) {$#1$};
    \node [style=\commandkey{style2}] (1) at (0, 0.5) {$#2$};
    \node [style=none] (2) at (0, 1.75) {};
  \end{pgfonlayer}
  \begin{pgfonlayer}{edgelayer}
    \draw (0) to (1);
    \draw (1) to (2.center);
  \end{pgfonlayer}
\end{tikzpicture}%
\,}

\newkeycommand{\kpointmapadj}[style1=kpointadj,style2=map][2]{\,%
\begin{tikzpicture}
  \begin{pgfonlayer}{nodelayer}
    \node [style=\commandkey{style1}] (0) at (0, 1.1) {$#1$};
    \node [style=\commandkey{style2}] (1) at (0, -0.5) {$#2$};
    \node [style=none] (2) at (0, -1.75) {};
  \end{pgfonlayer}
  \begin{pgfonlayer}{edgelayer}
    \draw (0) to (1);
    \draw (1) to (2.center);
  \end{pgfonlayer}
\end{tikzpicture}%
\,}

\newkeycommand{\dkpointmap}[style1=dkpoint,style2=dmap][2]{\,%
\begin{tikzpicture}
  \begin{pgfonlayer}{nodelayer}
    \node [style=\commandkey{style1}] (0) at (0, -1.2) {$#1$};
    \node [style=\commandkey{style2}] (1) at (0, 0.4) {$#2$};
    \node [style=none] (2) at (0, 1.7) {};
  \end{pgfonlayer}
  \begin{pgfonlayer}{edgelayer}
    \draw[style=boldedge] (0) to (1);
    \draw[style=boldedge] (1) to (2.center);
  \end{pgfonlayer}
\end{tikzpicture}%
\,}

\newkeycommand{\dkpointmapx}[style1=dkpoint,style2=dmap][2]{\,%
\begin{tikzpicture}
  \begin{pgfonlayer}{nodelayer}
    \node [style=\commandkey{style1}] (0) at (0, -1.4) {$#1$};
    \node [style=\commandkey{style2}] (1) at (0, 0.6) {$#2$};
    \node [style=none] (2) at (0, 1.9) {};
  \end{pgfonlayer}
  \begin{pgfonlayer}{edgelayer}
    \draw[style=boldedge] (0) to (1);
    \draw[style=boldedge] (1) to (2.center);
  \end{pgfonlayer}
\end{tikzpicture}%
\,}

\newkeycommand{\boxcopointmapdag}[style1=map,style2=copoint][2]{\,%
\begin{tikzpicture}
  \begin{pgfonlayer}{nodelayer}
    \node [style=none] (0) at (0, -1.75) {};
    \node [style=\commandkey{style1}] (1) at (0, -0.5) {$#1$};
    \node [style=\commandkey{style2}] (2) at (0, 1) {$#2$};
  \end{pgfonlayer}
  \begin{pgfonlayer}{edgelayer}
    \draw (0.center) to (1);
    \draw (1) to (2);
  \end{pgfonlayer}
\end{tikzpicture}%
\,}

\newkeycommand{\kpointdagmap}[style1=map,style2=kpointdag][2]{\,%
\begin{tikzpicture}
  \begin{pgfonlayer}{nodelayer}
    \node [style=none] (0) at (0, -1.75) {};
    \node [style=\commandkey{style1}] (1) at (0, -0.5) {$#1$};
    \node [style=\commandkey{style2}] (2) at (0, 1.1) {$#2$};
  \end{pgfonlayer}
  \begin{pgfonlayer}{edgelayer}
    \draw (0.center) to (1);
    \draw (1) to (2);
  \end{pgfonlayer}
\end{tikzpicture}%
\,}

\newkeycommand{\sandwichmap}[style1=point,style2=map,style3=copoint][3]{\,%
\begin{tikzpicture}
  \begin{pgfonlayer}{nodelayer}
    \node [style=\commandkey{style1}] (0) at (0, -1.50) {$#1$};
    \node [style=\commandkey{style2}] (1) at (0, 0) {$#2$};
    \node [style=\commandkey{style3}] (2) at (0, 1.50) {$#3$};
  \end{pgfonlayer}
  \begin{pgfonlayer}{edgelayer}
    \draw (0) to (1);
    \draw (1) to (2);
  \end{pgfonlayer}
\end{tikzpicture}%
}

\newkeycommand{\kpointsandwichmap}[style1=kpoint,style2=map,style3=kpointdag][3]{\,%
\begin{tikzpicture}
  \begin{pgfonlayer}{nodelayer}
    \node [style=\commandkey{style1}] (0) at (0, -1.5) {$#1$};
    \node [style=\commandkey{style2}] (1) at (0, 0) {$#2$};
    \node [style=\commandkey{style3}] (2) at (0, 1.5) {$#3$};
  \end{pgfonlayer}
  \begin{pgfonlayer}{edgelayer}
    \draw (0) to (1);
    \draw (1) to (2);
  \end{pgfonlayer}
\end{tikzpicture}%
}

\newkeycommand{\sandwichmapconj}[style1=point,style2=mapconj,style3=copoint][3]{\,%
\begin{tikzpicture}
  \begin{pgfonlayer}{nodelayer}
    \node [style=\commandkey{style1}] (0) at (0, -1.50) {$#1$};
    \node [style=\commandkey{style2}] (1) at (0, 0) {$#2$};
    \node [style=\commandkey{style3}] (2) at (0, 1.50) {$#3$};
  \end{pgfonlayer}
  \begin{pgfonlayer}{edgelayer}
    \draw (0) to (1);
    \draw (1) to (2);
  \end{pgfonlayer}
\end{tikzpicture}%
}

\newkeycommand{\sandwichtwo}[style1=point,style2=map,style3=map,style4=copoint][4]{\,%
\begin{tikzpicture}
  \begin{pgfonlayer}{nodelayer}
    \node [style=\commandkey{style2}] (0) at (0, -1) {$#2$};
    \node [style=\commandkey{style3}] (1) at (0, 1) {$#3$};
    \node [style=\commandkey{style1}] (2) at (0, -2.6) {$#1$};
    \node [style=\commandkey{style4}] (3) at (0, 2.6) {$#4$};
  \end{pgfonlayer}
  \begin{pgfonlayer}{edgelayer}
    \draw (2) to (0);
    \draw (0) to (1);
    \draw (1) to (3);
  \end{pgfonlayer}
\end{tikzpicture}\,}

\newkeycommand{\longbraket}[style1=point,style2=copoint][2]{\,%
\begin{tikzpicture}
  \begin{pgfonlayer}{nodelayer}
    \node [style=\commandkey{style1}] (0) at (0, -2.6) {$#1$};
    \node [style=\commandkey{style2}] (1) at (0, 2.6) {$#2$};
  \end{pgfonlayer}
  \begin{pgfonlayer}{edgelayer}
    \draw (0) to (1);
  \end{pgfonlayer}
\end{tikzpicture}\,}

\newkeycommand{\onb}[style=point]{\ensuremath{\{\,\tikz{\node[style=\commandkey{style}] (x) at (0,-0.3) {$j$};\draw (x)--(0,0.7);}\,\}}\xspace}

\newkeycommand{\phasebraket}[style1=gray phase dot,style2=white phase dot][2]{\begin{tikzpicture}
  \begin{pgfonlayer}{nodelayer}
    \node [style=\commandkey{style1}] (0) at (0, -0.75) {$#1$};
    \node [style=\commandkey{style2}] (1) at (0, 0.5) {$#2$};
  \end{pgfonlayer}
  \begin{pgfonlayer}{edgelayer}
    \draw (1) to (0);
  \end{pgfonlayer}
\end{tikzpicture}}

\newkeycommand{\phase}[style=white phase dot][1]{\,\begin{tikzpicture}
    \begin{pgfonlayer}{nodelayer}
        \node [style=none] (0) at (0, 1) {};
        \node [style=\commandkey{style}] (2) at (0, -0) {$#1$}; 
        \node [style=none] (3) at (0, -1) {};
    \end{pgfonlayer}
    \begin{pgfonlayer}{edgelayer}
        \draw (2) to (0.center);
        \draw (3.center) to (2);
    \end{pgfonlayer}
\end{tikzpicture}\,}

\newkeycommand{\dphase}[style=white phase ddot][1]{\,\begin{tikzpicture}
    \begin{pgfonlayer}{nodelayer}
        \node [style=none] (0) at (0, 1) {};
        \node [style=\commandkey{style}] (2) at (0, -0) {$#1$};
        \node [style=none] (3) at (0, -1) {};
    \end{pgfonlayer}
    \begin{pgfonlayer}{edgelayer}
        \draw [boldedge] (2) to (0.center);
        \draw [boldedge] (3.center) to (2);
    \end{pgfonlayer}
\end{tikzpicture}\,}

\newkeycommand{\dphasepoint}[style=white phase ddot][1]{\,\begin{tikzpicture}[yshift=-3mm]
  \begin{pgfonlayer}{nodelayer}
    \node [style=none] (0) at (0, 1) {};
    \node [style=\commandkey{style}] (1) at (0, 0) {$#1$};
  \end{pgfonlayer}
  \begin{pgfonlayer}{edgelayer}
    \draw [style=boldedge] (0.center) to (1);
  \end{pgfonlayer}
\end{tikzpicture}\,}

\newkeycommand{\dphasecopoint}[style=white phase ddot][1]{\,\begin{tikzpicture}[yshift=3mm,yscale=-1]
  \begin{pgfonlayer}{nodelayer}
    \node [style=none] (0) at (0, 1) {};
    \node [style=\commandkey{style}] (1) at (0, 0) {$#1$};
  \end{pgfonlayer}
  \begin{pgfonlayer}{edgelayer}
    \draw [style=boldedge] (0.center) to (1);
  \end{pgfonlayer}
\end{tikzpicture}\,}

\newkeycommand{\dphasepointsm}[style=white phase ddot][1]{\,\begin{tikzpicture}[yshift=-3mm]
  \begin{pgfonlayer}{nodelayer}
    \node [style=none] (0) at (0, 1) {};
    \node [style=\commandkey{style}] (1) at (0, 0) {\footnotesize\!$#1$\!};
  \end{pgfonlayer}
  \begin{pgfonlayer}{edgelayer}
    \draw [style=boldedge] (0.center) to (1);
  \end{pgfonlayer}
\end{tikzpicture}\,}

\newkeycommand{\phasepoint}[style=white phase dot][1]{\,\begin{tikzpicture}
  \begin{pgfonlayer}{nodelayer}
    \node [style=none] (0) at (0, 0.75) {};
    \node [style=\commandkey{style}] (1) at (0, -0.25) {$#1$};
  \end{pgfonlayer}
  \begin{pgfonlayer}{edgelayer}
    \draw (0.center) to (1);
  \end{pgfonlayer}
\end{tikzpicture}\,}

\newkeycommand{\phasecopoint}[style=white phase dot][1]{\,\begin{tikzpicture}[yshift=3mm]
  \begin{pgfonlayer}{nodelayer}
    \node [style=none] (0) at (0, -0.75) {};
    \node [style=\commandkey{style}] (1) at (0, 0.25) {$#1$};
  \end{pgfonlayer}
  \begin{pgfonlayer}{edgelayer}
    \draw (0.center) to (1);
  \end{pgfonlayer}
\end{tikzpicture}\,}

\newcommand{\innerprodmap}[2]{\,\tikz{
\node[style=copoint] (y) at (0,0.625) {$#1$};
\node[style=point] (x) at (0,-0.625) {$#2$};
\draw (x)--(y);}\,}

\newkeycommand{\kpointbraket}[style1=kpoint,style2=kpoint adjoint,estyle=][2]{\,%
\begin{tikzpicture}
  \begin{pgfonlayer}{nodelayer}
    \node [style=\commandkey{style1}] (0) at (0, -0.735) {$#1$};
    \node [style=\commandkey{style2}] (1) at (0, 0.735) {$#2$};
  \end{pgfonlayer}
  \begin{pgfonlayer}{edgelayer}
    \draw [\commandkey{estyle}] (0) to (1);
  \end{pgfonlayer}
\end{tikzpicture}\,}

\newkeycommand{\kkpointbraket}[style1=kpoint,style2=kpoint adjoint,estyle=][2]{\,%
\begin{tikzpicture}
  \begin{pgfonlayer}{nodelayer}
    \node [style=\commandkey{style1}] (0) at (0, -0.685) {$#1$};
    \node [style=\commandkey{style2}] (1) at (0, 0.685) {$#2$};
  \end{pgfonlayer}
  \begin{pgfonlayer}{edgelayer}
    \draw [\commandkey{estyle}] (0) to (1);
  \end{pgfonlayer}
\end{tikzpicture}\,}

\newkeycommand{\kpointbraketx}[style1=kpoint,style2=kpoint adjoint,estyle=][2]{\,%
\begin{tikzpicture}
  \begin{pgfonlayer}{nodelayer}
    \node [style=\commandkey{style1}] (0) at (0, -0.885) {$#1$};
    \node [style=\commandkey{style2}] (1) at (0, 0.885) {$#2$};
  \end{pgfonlayer}
  \begin{pgfonlayer}{edgelayer}
    \draw [\commandkey{estyle}] (0) to (1);
  \end{pgfonlayer}
\end{tikzpicture}\,}

\tikzstyle{cdiag}=[matrix of math nodes, row sep=3em, column sep=3em, text height=1.5ex, text depth=0.25ex,inner sep=0.5em]
\tikzstyle{arrow above}=[transform canvas={yshift=0.5ex}]
\tikzstyle{arrow below}=[transform canvas={yshift=-0.5ex}]

\tikzset{
  col/.cd,
  0/.style={draw,fill=white}
  1/.style={draw,fill=blue}
}

\tikzstyle{B}=[draw,fill=gray!90!black]
\tikzstyle{W}=[draw,fill=white]
\tikzstyle{G}=[draw,fill=gray!50!white]

\tikzstyle{env}=[copoint,regular polygon rotate=0,minimum width=0.2cm, fill=black]

\tikzstyle{probs}=[shape=semicircle,fill=white,draw=black,shape border rotate=180,minimum width=1.2cm]

\tikzstyle{every picture}=[baseline=-0.25em,scale=0.5]
\tikzstyle{dotpic}=[] 
\tikzstyle{diredges}=[every to/.style={diredge}]
\tikzstyle{math matrix}=[matrix of math nodes,left delimiter=(,right delimiter=),inner sep=2pt,column sep=1em,row sep=0.5em,nodes={inner sep=0pt},text height=1.5ex, text depth=0.25ex]

\tikzstyle{inline text}=[text height=1.5ex, text depth=0.25ex,yshift=0.5mm]
\tikzstyle{label}=[font=\footnotesize,text height=1.5ex, text depth=0.25ex,yshift=0.5mm]
\tikzstyle{left label}=[label,anchor=east,xshift=1.5mm]
\tikzstyle{right label}=[label,anchor=west,xshift=-1.5mm]

\tikzstyle{braceedge}=[decorate,decoration={brace,amplitude=2mm,raise=-1mm}]
\tikzstyle{small braceedge}=[decorate,decoration={brace,amplitude=1mm,raise=-1mm}]

\tikzstyle{doubled}=[line width=1.6pt] 
\tikzstyle{boldedge}=[doubled,shorten <=-0.17mm,shorten >=-0.17mm]
\tikzstyle{boldedgegray}=[doubled,gray,shorten <=-0.17mm,shorten >=-0.17mm]
\tikzstyle{singleedgegray}=[gray]

\tikzstyle{semidoubled}=[line width=1.4pt] 
\tikzstyle{semiboldedgegray}=[semidoubled,gray,shorten <=-0.17mm,shorten >=-0.17mm]

\tikzstyle{boxedge}=[semiboldedgegray]

\tikzstyle{boldedgedashed}=[very thick,dashed,shorten <=-0.17mm,shorten >=-0.17mm]
\tikzstyle{vboldedgedashed}=[doubled,dashed,shorten <=-0.17mm,shorten >=-0.17mm]
\tikzstyle{left hook arrow}=[left hook-latex]
\tikzstyle{right hook arrow}=[right hook-latex]
\tikzstyle{sembracket}=[line width=0.5pt,shorten <=-0.07mm,shorten >=-0.07mm]

\tikzstyle{causal edge}=[->,thick,gray]
\tikzstyle{causal nondir}=[thick,gray]
\tikzstyle{timeline}=[thick,gray, dashed]

\tikzstyle{cedge}=[<->,thick,gray!70!white]

\tikzstyle{empty diagram}=[draw=gray!40!white,dashed,shape=rectangle,minimum width=1cm,minimum height=1cm]
\tikzstyle{empty diagram small}=[draw=gray!50!white,dashed,shape=rectangle,minimum width=0.6cm,minimum height=0.5cm]

\tikzstyle{dot}=[inner sep=0mm,minimum width=2mm,minimum height=2mm,draw,shape=circle]  
\tikzstyle{Wsquare}=[white dot, shape=regular polygon, rounded corners=0.8 mm, minimum size=3.3 mm, regular polygon sides=3, outer sep=-0.2mm]
\tikzstyle{Wsquareadj}=[white dot, shape=regular polygon, rounded corners=0.8 mm, minimum size=3.3 mm, regular polygon sides=3, outer sep=-0.2mm, regular polygon rotate=180]
\tikzstyle{ddot}=[inner sep=0mm, doubled, minimum width=2.5mm,minimum height=2.5mm,draw,shape=circle]

\tikzstyle{black dot}=[dot,fill=black]
\tikzstyle{white dot}=[dot,fill=white,,text depth=-0.2mm]
\tikzstyle{white Wsquare}=[Wsquare,fill=gray,,text depth=-0.2mm]
\tikzstyle{white Wsquareadj}=[Wsquareadj,fill=white,,text depth=-0.2mm]
\tikzstyle{green dot}=[white dot] 
\tikzstyle{gray dot}=[dot,fill=gray!40!white,,text depth=-0.2mm]
\tikzstyle{red dot}=[gray dot] 

\tikzstyle{black ddot}=[ddot,fill=black]
\tikzstyle{white ddot}=[ddot,fill=white]
\tikzstyle{gray ddot}=[ddot,fill=gray!40!white]

\tikzstyle{gray edge}=[gray!60!white]

\tikzstyle{small dot}=[inner sep=0.5mm,minimum width=0pt,minimum height=0pt,draw,shape=circle]

\tikzstyle{small black dot}=[small dot,fill=black]
\tikzstyle{small white dot}=[small dot,fill=white]
\tikzstyle{small gray dot}=[small dot,fill=gray!40!white]

\tikzstyle{causal dot}=[inner sep=0.4mm,minimum width=0pt,minimum height=0pt,draw=white,shape=circle,fill=gray!40!white]

\tikzstyle{phase dimensions}=[minimum size=5mm,font=\footnotesize,rectangle,rounded corners=2.5mm,inner sep=0.2mm,outer sep=-2mm]
\tikzstyle{dphase dimensions}=[minimum size=5mm,font=\footnotesize,rectangle,rounded corners=2.5mm,inner sep=0.2mm,outer sep=-2mm]

\tikzstyle{white phase dot}=[dot,fill=white,phase dimensions]
\tikzstyle{white phase ddot}=[ddot,fill=white,dphase dimensions]

\tikzstyle{white rect ddot}=[draw=black,fill=white,doubled,minimum size=5mm,font=\footnotesize,rectangle,rounded corners=2.5mm,inner sep=0.2mm]
\tikzstyle{gray rect ddot}=[draw=black,fill=gray!40!white,doubled,minimum size=6mm,font=\footnotesize,rectangle,rounded corners=3mm]

\tikzstyle{gray phase dot}=[dot,fill=gray!40!white,phase dimensions]
\tikzstyle{gray phase ddot}=[ddot,fill=gray!40!white,dphase dimensions]
\tikzstyle{grey phase dot}=[gray phase dot]
\tikzstyle{grey phase ddot}=[gray phase ddot]

\tikzstyle{small phase dimensions}=[minimum size=4mm,font=\tiny,rectangle,rounded corners=2mm,inner sep=0.2mm,outer sep=-2mm]
\tikzstyle{small dphase dimensions}=[minimum size=4mm,font=\tiny,rectangle,rounded corners=2mm,inner sep=0.2mm,outer sep=-2mm]

\tikzstyle{small gray phase dot}=[dot,fill=gray!40!white,small phase dimensions]
\tikzstyle{small gray phase ddot}=[ddot,fill=gray!40!white,small dphase dimensions]

\tikzstyle{small map}=[draw,shape=rectangle,minimum height=4mm,minimum width=4mm,fill=white]

\tikzstyle{cnot}=[fill=white,shape=circle,inner sep=-1.4pt]

\tikzstyle{asym hadamard}=[fill=white,draw,shape=NEbox,inner sep=0.6mm,font=\footnotesize,minimum height=4mm]
\tikzstyle{asym hadamard conj}=[fill=white,draw,shape=NWbox,inner sep=0.6mm,font=\footnotesize,minimum height=4mm]
\tikzstyle{asym hadamard dag}=[fill=white,draw,shape=SEbox,inner sep=0.6mm,font=\footnotesize,minimum height=4mm]

\tikzstyle{hadamard}=[fill=white,draw,inner sep=0.6mm,font=\footnotesize,minimum height=4mm,minimum width=4mm]
\tikzstyle{small hadamard}=[fill=white,draw,inner sep=0.6mm,minimum height=1.5mm,minimum width=1.5mm]
\tikzstyle{small hadamard rotate}=[small hadamard,rotate=45]
\tikzstyle{dhadamard}=[hadamard,doubled]
\tikzstyle{small dhadamard}=[small hadamard,doubled]
\tikzstyle{small dhadamard rotate}=[small hadamard rotate,doubled]
\tikzstyle{antipode}=[white dot,inner sep=0.3mm,font=\footnotesize]

\tikzstyle{scalar}=[diamond,draw,inner sep=0.5pt,font=\small]
\tikzstyle{dscalar}=[diamond,doubled, draw,inner sep=0.5pt,font=\small]

\tikzstyle{small box}=[rectangle,inline text,fill=white,draw,minimum height=5mm,yshift=-0.5mm,minimum width=5mm,font=\small]
\tikzstyle{small gray box}=[small box,fill=gray!30]
\tikzstyle{medium box}=[rectangle,inline text,fill=white,draw,minimum height=5mm,yshift=-0.5mm,minimum width=10mm,font=\small]
\tikzstyle{square box}=[small box] 
\tikzstyle{medium gray box}=[small box,fill=gray!30]
\tikzstyle{semilarge box}=[rectangle,inline text,fill=white,draw,minimum height=5mm,yshift=-0.5mm,minimum width=12.5mm,font=\small]
\tikzstyle{large box}=[rectangle,inline text,fill=white,draw,minimum height=5mm,yshift=-0.5mm,minimum width=15mm,font=\small]
\tikzstyle{large gray box}=[small box,fill=gray!30]

\tikzstyle{Bayes box}=[rectangle,fill=black,draw, minimum height=3mm, minimum width=3mm]

\tikzstyle{gray square point}=[small box,fill=gray!50]

\tikzstyle{dphase box white}=[dhadamard]
\tikzstyle{dphase box gray}=[dhadamard,fill=gray!50!white]
\tikzstyle{phase box white}=[hadamard]
\tikzstyle{phase box gray}=[hadamard,fill=gray!50!white]

\tikzstyle{point}=[regular polygon,regular polygon sides=3,draw,scale=0.75,inner sep=-0.5pt,minimum width=9mm,fill=white,regular polygon rotate=180]
\tikzstyle{point nosep}=[regular polygon,regular polygon sides=3,draw,scale=0.75,inner sep=-2pt,minimum width=9mm,fill=white,regular polygon rotate=180]
\tikzstyle{copoint}=[regular polygon,regular polygon sides=3,draw,scale=0.75,inner sep=-0.5pt,minimum width=9mm,fill=white]
\tikzstyle{dpoint}=[point,doubled]
\tikzstyle{dcopoint}=[copoint,doubled]

\tikzstyle{pointgrow}=[shape=cornerpoint,kpoint common,scale=0.75,inner sep=3pt]
\tikzstyle{pointgrow dag}=[shape=cornercopoint,kpoint common,scale=0.75,inner sep=3pt]

\tikzstyle{wide copoint}=[fill=white,draw,shape=isosceles triangle,shape border rotate=90,isosceles triangle stretches=true,inner sep=0pt,minimum width=1.5cm,minimum height=6.12mm]
\tikzstyle{wide point}=[fill=white,draw,shape=isosceles triangle,shape border rotate=-90,isosceles triangle stretches=true,inner sep=0pt,minimum width=1.5cm,minimum height=6.12mm,yshift=-0.0mm]
\tikzstyle{wide point plus}=[fill=white,draw,shape=isosceles triangle,shape border rotate=-90,isosceles triangle stretches=true,inner sep=0pt,minimum width=1.74cm,minimum height=7mm,yshift=-0.0mm]

\tikzstyle{wide dpoint}=[fill=white,doubled,draw,shape=isosceles triangle,shape border rotate=-90,isosceles triangle stretches=true,inner sep=0pt,minimum width=1.5cm,minimum height=6.12mm,yshift=-0.0mm]

\tikzstyle{tinypoint}=[regular polygon,regular polygon sides=3,draw,scale=0.55,inner sep=-0.15pt,minimum width=6mm,fill=white,regular polygon rotate=180] 

\tikzstyle{white point}=[point]
\tikzstyle{white dpoint}=[dpoint]
\tikzstyle{green point}=[white point] 
\tikzstyle{white copoint}=[copoint]
\tikzstyle{gray point}=[point,fill=gray!40!white]
\tikzstyle{gray dpoint}=[gray point,doubled]
\tikzstyle{red point}=[gray point] 
\tikzstyle{gray copoint}=[copoint,fill=gray!40!white]
\tikzstyle{gray dcopoint}=[gray copoint,doubled]

\tikzstyle{white point guide}=[regular polygon,regular polygon sides=3,font=\scriptsize,draw,scale=0.65,inner sep=-0.5pt,minimum width=9mm,fill=white,regular polygon rotate=180]

\tikzstyle{black point}=[point,fill=black,font=\color{white}]
\tikzstyle{black copoint}=[copoint,fill=black,font=\color{white}]

\tikzstyle{tiny gray point}=[tinypoint,fill=gray!40!white]

\tikzstyle{diredge}=[->]
\tikzstyle{ddiredge}=[<->]
\tikzstyle{rdiredge}=[<-]
\tikzstyle{thickdiredge}=[->, very thick]
\tikzstyle{pointer edge}=[->,very thick,gray]
\tikzstyle{pointer edge part}=[very thick,gray]
\tikzstyle{dashed edge}=[dashed]
\tikzstyle{thick dashed edge}=[very thick,dashed]
\tikzstyle{thick gray dashed edge}=[thick dashed edge,gray!40]
\tikzstyle{thick map edge}=[very thick,|->]

\makeatletter
\newcommand{\boxshape}[3]{%
\pgfdeclareshape{#1}{
\inheritsavedanchors[from=rectangle] 
\inheritanchorborder[from=rectangle]
\inheritanchor[from=rectangle]{center}
\inheritanchor[from=rectangle]{north}
\inheritanchor[from=rectangle]{south}
\inheritanchor[from=rectangle]{west}
\inheritanchor[from=rectangle]{east}
\backgroundpath{
\southwest \pgf@xa=\pgf@x \pgf@ya=\pgf@y
\northeast \pgf@xb=\pgf@x \pgf@yb=\pgf@y

\@tempdima=#2
\@tempdimb=#3

\pgfpathmoveto{\pgfpoint{\pgf@xa - 5pt + \@tempdima}{\pgf@ya}}
\pgfpathlineto{\pgfpoint{\pgf@xa - 5pt - \@tempdima}{\pgf@yb}}
\pgfpathlineto{\pgfpoint{\pgf@xb + 5pt + \@tempdimb}{\pgf@yb}}
\pgfpathlineto{\pgfpoint{\pgf@xb + 5pt - \@tempdimb}{\pgf@ya}}
\pgfpathlineto{\pgfpoint{\pgf@xa - 5pt + \@tempdima}{\pgf@ya}}
\pgfpathclose
}
}}

\boxshape{NEbox}{0pt}{5pt}
\boxshape{SEbox}{0pt}{-5pt}
\boxshape{NWbox}{5pt}{0pt}
\boxshape{SWbox}{-5pt}{0pt}
\boxshape{EBox}{-3pt}{3pt}
\boxshape{WBox}{3pt}{-3pt}
\makeatother

\tikzstyle{cloud}=[shape=cloud,draw,minimum width=1.5cm,minimum height=1.5cm]

\tikzstyle{map}=[draw,shape=NEbox,inner sep=2pt,minimum height=6mm,fill=white]
\tikzstyle{dashedmap}=[draw,dashed,shape=NEbox,inner sep=2pt,minimum height=6mm,fill=white]
\tikzstyle{mapdag}=[draw,shape=SEbox,inner sep=2pt,minimum height=6mm,fill=white]
\tikzstyle{mapadj}=[draw,shape=SEbox,inner sep=2pt,minimum height=6mm,fill=white]
\tikzstyle{maptrans}=[draw,shape=SWbox,inner sep=2pt,minimum height=6mm,fill=white]
\tikzstyle{mapconj}=[draw,shape=NWbox,inner sep=2pt,minimum height=6mm,fill=white]

\tikzstyle{medium map}=[draw,shape=NEbox,inner sep=2pt,minimum height=6mm,fill=white,minimum width=7mm]
\tikzstyle{medium map dag}=[draw,shape=SEbox,inner sep=2pt,minimum height=6mm,fill=white,minimum width=7mm]
\tikzstyle{medium map adj}=[draw,shape=SEbox,inner sep=2pt,minimum height=6mm,fill=white,minimum width=7mm]
\tikzstyle{medium map trans}=[draw,shape=SWbox,inner sep=2pt,minimum height=6mm,fill=white,minimum width=7mm]
\tikzstyle{medium map conj}=[draw,shape=NWbox,inner sep=2pt,minimum height=6mm,fill=white,minimum width=7mm]
\tikzstyle{semilarge map}=[draw,shape=NEbox,inner sep=2pt,minimum height=6mm,fill=white,minimum width=9.5mm]
\tikzstyle{semilarge map trans}=[draw,shape=SWbox,inner sep=2pt,minimum height=6mm,fill=white,minimum width=9.5mm]
\tikzstyle{semilarge map adj}=[draw,shape=SEbox,inner sep=2pt,minimum height=6mm,fill=white,minimum width=9.5mm]
\tikzstyle{semilarge map dag}=[draw,shape=SEbox,inner sep=2pt,minimum height=6mm,fill=white,minimum width=9.5mm]
\tikzstyle{semilarge map conj}=[draw,shape=NWbox,inner sep=2pt,minimum height=6mm,fill=white,minimum width=9.5mm]
\tikzstyle{large map}=[draw,shape=NEbox,inner sep=2pt,minimum height=6mm,fill=white,minimum width=12mm]
\tikzstyle{large map conj}=[draw,shape=NWbox,inner sep=2pt,minimum height=6mm,fill=white,minimum width=12mm]
\tikzstyle{very large map}=[draw,shape=NEbox,inner sep=2pt,minimum height=6mm,fill=white,minimum width=17mm]

\tikzstyle{medium dmap}=[draw,doubled,shape=NEbox,inner sep=2pt,minimum height=6mm,fill=white,minimum width=7mm]
\tikzstyle{medium dmap dag}=[draw,doubled,shape=SEbox,inner sep=2pt,minimum height=6mm,fill=white,minimum width=7mm]
\tikzstyle{medium dmap adj}=[draw,doubled,shape=SEbox,inner sep=2pt,minimum height=6mm,fill=white,minimum width=7mm]
\tikzstyle{medium dmap trans}=[draw,doubled,shape=SWbox,inner sep=2pt,minimum height=6mm,fill=white,minimum width=7mm]
\tikzstyle{medium dmap conj}=[draw,doubled,shape=NWbox,inner sep=2pt,minimum height=6mm,fill=white,minimum width=7mm]
\tikzstyle{semilarge dmap}=[draw,doubled,shape=NEbox,inner sep=2pt,minimum height=6mm,fill=white,minimum width=9.5mm]
\tikzstyle{semilarge dmap trans}=[draw,doubled,shape=SWbox,inner sep=2pt,minimum height=6mm,fill=white,minimum width=9.5mm]
\tikzstyle{semilarge dmap adj}=[draw,doubled,shape=SEbox,inner sep=2pt,minimum height=6mm,fill=white,minimum width=9.5mm]
\tikzstyle{semilarge dmap dag}=[draw,doubled,shape=SEbox,inner sep=2pt,minimum height=6mm,fill=white,minimum width=9.5mm]
\tikzstyle{semilarge dmap conj}=[draw,doubled,shape=NWbox,inner sep=2pt,minimum height=6mm,fill=white,minimum width=9.5mm]
\tikzstyle{large dmap}=[draw,doubled,shape=NEbox,inner sep=2pt,minimum height=6mm,fill=white,minimum width=12mm]
\tikzstyle{large dmap conj}=[draw,doubled,shape=NWbox,inner sep=2pt,minimum height=6mm,fill=white,minimum width=12mm]
\tikzstyle{large dmap trans}=[draw,doubled,shape=SWbox,inner sep=2pt,minimum height=6mm,fill=white,minimum width=12mm]
\tikzstyle{large dmap adj}=[draw,doubled,shape=SEbox,inner sep=2pt,minimum height=6mm,fill=white,minimum width=12mm]
\tikzstyle{large dmap dag}=[draw,doubled,shape=SEbox,inner sep=2pt,minimum height=6mm,fill=white,minimum width=12mm]
\tikzstyle{very large dmap}=[draw,doubled,shape=NEbox,inner sep=2pt,minimum height=6mm,fill=white,minimum width=19.5mm]

\tikzstyle{muxbox}=[draw,shape=rectangle,minimum height=3mm,minimum width=3mm,fill=white]
\tikzstyle{dmuxbox}=[muxbox,doubled]

\tikzstyle{box}=[draw,shape=rectangle,inner sep=2pt,minimum height=6mm,minimum width=6mm,fill=white]
\tikzstyle{dbox}=[draw,doubled,shape=rectangle,inner sep=2pt,minimum height=6mm,minimum width=6mm,fill=white]
\tikzstyle{dmap}=[draw,doubled,shape=NEbox,inner sep=2pt,minimum height=6mm,fill=white]
\tikzstyle{dmapdag}=[draw,doubled,shape=SEbox,inner sep=2pt,minimum height=6mm,fill=white]
\tikzstyle{dmapadj}=[draw,doubled,shape=SEbox,inner sep=2pt,minimum height=6mm,fill=white]
\tikzstyle{dmaptrans}=[draw,doubled,shape=SWbox,inner sep=2pt,minimum height=6mm,fill=white]
\tikzstyle{dmapconj}=[draw,doubled,shape=NWbox,inner sep=2pt,minimum height=6mm,fill=white]

\tikzstyle{ddmap}=[draw,doubled,dashed,shape=NEbox,inner sep=2pt,minimum height=6mm,fill=white]
\tikzstyle{ddmapdag}=[draw,doubled,dashed,shape=SEbox,inner sep=2pt,minimum height=6mm,fill=white]
\tikzstyle{ddmapadj}=[draw,doubled,dashed,shape=SEbox,inner sep=2pt,minimum height=6mm,fill=white]
\tikzstyle{ddmaptrans}=[draw,doubled,dashed,shape=SWbox,inner sep=2pt,minimum height=6mm,fill=white]
\tikzstyle{ddmapconj}=[draw,doubled,dashed,shape=NWbox,inner sep=2pt,minimum height=6mm,fill=white]

\boxshape{sNEbox}{0pt}{3pt}
\boxshape{sSEbox}{0pt}{-3pt}
\boxshape{sNWbox}{3pt}{0pt}
\boxshape{sSWbox}{-3pt}{0pt}
\tikzstyle{smap}=[draw,shape=sNEbox,fill=white]
\tikzstyle{smapdag}=[draw,shape=sSEbox,fill=white]
\tikzstyle{smapadj}=[draw,shape=sSEbox,fill=white]
\tikzstyle{smaptrans}=[draw,shape=sSWbox,fill=white]
\tikzstyle{smapconj}=[draw,shape=sNWbox,fill=white]

\tikzstyle{dsmap}=[draw,dashed,shape=sNEbox,fill=white]
\tikzstyle{dsmapdag}=[draw,dashed,shape=sSEbox,fill=white]
\tikzstyle{dsmaptrans}=[draw,dashed,shape=sSWbox,fill=white]
\tikzstyle{dsmapconj}=[draw,dashed,shape=sNWbox,fill=white]

\boxshape{mNEbox}{0pt}{10pt}
\boxshape{mSEbox}{0pt}{-10pt}
\boxshape{mNWbox}{10pt}{0pt}
\boxshape{mSWbox}{-10pt}{0pt}
\tikzstyle{mmap}=[draw,shape=mNEbox]
\tikzstyle{mmapdag}=[draw,shape=mSEbox]
\tikzstyle{mmaptrans}=[draw,shape=mSWbox]
\tikzstyle{mmapconj}=[draw,shape=mNWbox]

\tikzstyle{mmapgray}=[draw,fill=gray!40!white,shape=mNEbox]
\tikzstyle{smapgray}=[draw,fill=gray!40!white,shape=sNEbox]

\makeatletter

\pgfdeclareshape{cornerpoint}{
\inheritsavedanchors[from=rectangle] 
\inheritanchorborder[from=rectangle]
\inheritanchor[from=rectangle]{center}
\inheritanchor[from=rectangle]{north}
\inheritanchor[from=rectangle]{south}
\inheritanchor[from=rectangle]{west}
\inheritanchor[from=rectangle]{east}
\backgroundpath{
\southwest \pgf@xa=\pgf@x \pgf@ya=\pgf@y
\northeast \pgf@xb=\pgf@x \pgf@yb=\pgf@y

\pgfmathsetmacro{\pgf@shorten@left}{\pgfkeysvalueof{/tikz/shorten left}}
\pgfmathsetmacro{\pgf@shorten@right}{\pgfkeysvalueof{/tikz/shorten right}}

\pgfpathmoveto{\pgfpoint{0.5 * (\pgf@xa + \pgf@xb)}{\pgf@ya - 5pt}}
\pgfpathlineto{\pgfpoint{\pgf@xa - 8pt + \pgf@shorten@left}{\pgf@yb - 1.5 * \pgf@shorten@left}}
\pgfpathlineto{\pgfpoint{\pgf@xa - 8pt + \pgf@shorten@left}{\pgf@yb}}
\pgfpathlineto{\pgfpoint{\pgf@xb + 8pt - \pgf@shorten@right}{\pgf@yb}}
\pgfpathlineto{\pgfpoint{\pgf@xb + 8pt - \pgf@shorten@right}{\pgf@yb - 1.5 * \pgf@shorten@right}}
\pgfpathclose
}
}

\pgfdeclareshape{cornercopoint}{
\inheritsavedanchors[from=rectangle] 
\inheritanchorborder[from=rectangle]
\inheritanchor[from=rectangle]{center}
\inheritanchor[from=rectangle]{north}
\inheritanchor[from=rectangle]{south}
\inheritanchor[from=rectangle]{west}
\inheritanchor[from=rectangle]{east}
\backgroundpath{
\southwest \pgf@xa=\pgf@x \pgf@ya=\pgf@y
\northeast \pgf@xb=\pgf@x \pgf@yb=\pgf@y

\pgfmathsetmacro{\pgf@shorten@left}{\pgfkeysvalueof{/tikz/shorten left}}
\pgfmathsetmacro{\pgf@shorten@right}{\pgfkeysvalueof{/tikz/shorten right}}

\pgfpathmoveto{\pgfpoint{0.5 * (\pgf@xa + \pgf@xb)}{\pgf@yb + 5pt}}
\pgfpathlineto{\pgfpoint{\pgf@xa - 8pt + \pgf@shorten@left}{\pgf@ya + 1.5 * \pgf@shorten@left}}
\pgfpathlineto{\pgfpoint{\pgf@xa - 8pt + \pgf@shorten@left}{\pgf@ya}}
\pgfpathlineto{\pgfpoint{\pgf@xb + 8pt - \pgf@shorten@right}{\pgf@ya}}
\pgfpathlineto{\pgfpoint{\pgf@xb + 8pt - \pgf@shorten@right}{\pgf@ya + 1.5 * \pgf@shorten@right}}
\pgfpathclose
}
}

\makeatother

\pgfkeyssetvalue{/tikz/shorten left}{0pt}
\pgfkeyssetvalue{/tikz/shorten right}{0pt}

\tikzstyle{kpoint common}=[draw,fill=white,inner sep=1pt,minimum height=4mm]
\tikzstyle{kpoint sc}=[shape=cornerpoint,kpoint common]
\tikzstyle{kpoint adjoint sc}=[shape=cornercopoint,kpoint common]
\tikzstyle{kpoint}=[shape=cornerpoint,shorten left=5pt,kpoint common]
\tikzstyle{kpoint adjoint}=[shape=cornercopoint,shorten left=5pt,kpoint common]
\tikzstyle{kpoint conjugate}=[shape=cornerpoint,shorten right=5pt,kpoint common]
\tikzstyle{kpoint transpose}=[shape=cornercopoint,shorten right=5pt,kpoint common]
\tikzstyle{kpoint symm}=[shape=cornerpoint,shorten left=5pt,shorten right=5pt,kpoint common]

\tikzstyle{wide kpoint sc}=[shape=cornerpoint,kpoint common, minimum width=1 cm]
\tikzstyle{wide kpointdag sc}=[shape=cornercopoint,kpoint common, minimum width=1 cm]

\tikzstyle{black kpoint}=[shape=cornerpoint,shorten left=5pt,kpoint common,fill=black,font=\color{white}]

\tikzstyle{black kpoint sm}=[shape=cornerpoint,shorten left=5pt,kpoint common,fill=black,font=\color{white},scale=0.75]

\tikzstyle{black kpoint adjoint}=[shape=cornercopoint,shorten left=5pt,kpoint common,fill=black,font=\color{white}]
\tikzstyle{black kpointadj}=[shape=cornercopoint,shorten left=5pt,kpoint common,fill=black,font=\color{white}]

\tikzstyle{black kpointadj sm}=[shape=cornercopoint,shorten left=5pt,kpoint common,fill=black,font=\color{white},scale=0.75]

\tikzstyle{black dkpoint}=[shape=cornerpoint,shorten left=5pt,kpoint common,fill=black, doubled,font=\color{white}]
\tikzstyle{black dkpoint adjoint}=[shape=cornercopoint,shorten left=5pt,kpoint common,fill=black, doubled,font=\color{white}]
\tikzstyle{black dkpointadj}=[shape=cornercopoint,shorten left=5pt,kpoint common,fill=black, doubled,font=\color{white}]

\tikzstyle{black dkpoint sm}=[shape=cornerpoint,shorten left=5pt,kpoint common,fill=black, doubled,font=\color{white},scale=0.75]
\tikzstyle{black dkpointadj sm}=[shape=cornercopoint,shorten left=5pt,kpoint common,fill=black, doubled,font=\color{white},scale=0.75] 

\tikzstyle{kpointdag}=[kpoint adjoint]
\tikzstyle{kpointadj}=[kpoint adjoint]
\tikzstyle{kpointconj}=[kpoint conjugate]
\tikzstyle{kpointtrans}=[kpoint transpose]

\tikzstyle{big kpoint}=[kpoint, minimum width=1.2 cm, minimum height=8mm, inner sep=4pt, text depth=3mm]

\tikzstyle{wide kpoint}=[kpoint, minimum width=1 cm, inner sep=2pt]
\tikzstyle{wide kpointdag}=[kpointdag, minimum width=1 cm, inner sep=2pt]
\tikzstyle{wide kpointconj}=[kpointconj, minimum width=1 cm, inner sep=2pt]
\tikzstyle{wide kpointtrans}=[kpointtrans, minimum width=1 cm, inner sep=2pt]

\tikzstyle{wider kpoint}=[kpoint, minimum width=1.25 cm, inner sep=2pt]
\tikzstyle{wider kpointdag}=[kpointdag, minimum width=1.25 cm, inner sep=2pt]
\tikzstyle{wider kpointconj}=[kpointconj, minimum width=1.25 cm, inner sep=2pt]
\tikzstyle{wider kpointtrans}=[kpointtrans, minimum width=1.25 cm, inner sep=2pt]

\tikzstyle{gray kpoint}=[kpoint,fill=gray!50!white]
\tikzstyle{gray kpointdag}=[kpointdag,fill=gray!50!white]
\tikzstyle{gray kpointadj}=[kpointadj,fill=gray!50!white]
\tikzstyle{gray kpointconj}=[kpointconj,fill=gray!50!white]
\tikzstyle{gray kpointtrans}=[kpointtrans,fill=gray!50!white]

\tikzstyle{gray dkpoint}=[kpoint,fill=gray!50!white,doubled]
\tikzstyle{gray dkpointdag}=[kpointdag,fill=gray!50!white,doubled]
\tikzstyle{gray dkpointadj}=[kpointadj,fill=gray!50!white,doubled]
\tikzstyle{gray dkpointconj}=[kpointconj,fill=gray!50!white,doubled]
\tikzstyle{gray dkpointtrans}=[kpointtrans,fill=gray!50!white,doubled]

\tikzstyle{white label}=[draw,fill=white,rectangle,inner sep=0.7 mm]
\tikzstyle{gray label}=[draw,fill=gray!50!white,rectangle,inner sep=0.7 mm]
\tikzstyle{black label}=[draw,fill=black,rectangle,inner sep=0.7 mm]

\tikzstyle{dkpoint}=[kpoint,doubled]
\tikzstyle{wide dkpoint}=[wide kpoint,doubled]
\tikzstyle{dkpointdag}=[kpoint adjoint,doubled]
\tikzstyle{wide dkpointdag}=[wide kpointdag,doubled]
\tikzstyle{dkcopoint}=[kpoint adjoint,doubled]
\tikzstyle{dkpointadj}=[kpoint adjoint,doubled]
\tikzstyle{dkpointconj}=[kpoint conjugate,doubled]
\tikzstyle{dkpointtrans}=[kpoint transpose,doubled]

\tikzstyle{kscalar}=[kpoint common, shape=EBox, inner xsep=-1pt, inner ysep=3pt,font=\small]
\tikzstyle{kscalarconj}=[kpoint common, shape=WBox, inner xsep=-1pt, inner ysep=3pt,font=\small]

\tikzstyle{spekpoint}=[kpoint sc,minimum height=5mm,inner sep=3pt]
\tikzstyle{spekcopoint}=[kpoint adjoint sc,minimum height=5mm,inner sep=3pt]

\tikzstyle{dspekpoint}=[spekpoint,doubled]
\tikzstyle{dspekcopoint}=[spekcopoint,doubled]

 \tikzstyle{upground}=[circuit ee IEC,thick,ground,rotate=90,scale=2.5]
 \tikzstyle{downground}=[circuit ee IEC,thick,ground,rotate=-90,scale=2.5]
 \tikzstyle{bigground}=[regular polygon,regular polygon sides=3,draw=gray,scale=0.50,inner sep=-0.5pt,minimum width=10mm,fill=gray]

\tikzstyle{arrs}=[-latex,font=\small,auto]
\tikzstyle{arrow plain}=[arrs]
\tikzstyle{arrow dashed}=[dashed,arrs]
\tikzstyle{arrow bold}=[very thick,arrs]
\tikzstyle{arrow hide}=[draw=white!0,-]
\tikzstyle{arrow reverse}=[latex-]
\tikzstyle{cdnode}=[]

\newcommand{\bigunit}[1]{%
\,\begin{tikzpicture}[dotpic,scale=2,yshift=1.5mm]
\node [#1] (a) at (0,-0.25) {}; 
\draw (a)--(0,0.2);
\end{tikzpicture}\,}

\newcommand{\dotcounit}[1]{%
\,\begin{tikzpicture}[dotpic,yshift=-1mm]
\node [#1] (a) at (0,0.35) {}; 
\draw (0,-0.3)--(a);
\end{tikzpicture}\,}
\newcommand{\dotunit}[1]{%
\,\begin{tikzpicture}[dotpic,yshift=1.5mm]
\node [#1] (a) at (0,-0.35) {}; 
\draw (a)--(0,0.3);
\end{tikzpicture}\,}
\newcommand{\dotcomult}[1]{%
\,\begin{tikzpicture}[dotpic,yshift=0.5mm]
	\node [#1] (a) {};
	\draw (-90:0.55)--(a);
	\draw (a) -- (45:0.6);
	\draw (a) -- (135:0.6);
\end{tikzpicture}\,}
\newcommand{\dotmult}[1]{%
\,\begin{tikzpicture}[dotpic]
	\node [#1] (a) {};
	\draw (a) -- (90:0.55);
	\draw (a) (-45:0.6) -- (a);
	\draw (a) (-135:0.6) -- (a);
\end{tikzpicture}\,}

\newcommand{\dotonly}[1]{%
\,\begin{tikzpicture}[dotpic]
\node [#1] (a) at (0,0) {};
\end{tikzpicture}\,}

\newcommand{\blackdot}{\dotonly{black dot}\xspace}

\newcommand{\whitedot}{\dotonly{white dot}\xspace}

\newcommand{\whiteunit}{\dotunit{white dot}\xspace}
\newcommand{\whitecounit}{\dotcounit{white dot}\xspace}
\newcommand{\whitemult}{\dotmult{white dot}\xspace}
\newcommand{\whitecomult}{\dotcomult{white dot}\xspace}

\newcommand{\wdot}{\dotonly{white Wsquare}\xspace}

\newcommand{\graymult}{\dotmult{gray dot}\xspace}

\let\olddagger\dagger
\renewcommand{\dagger}{\ensuremath{\olddagger}\xspace}

\usepackage{makeidx}
\makeindex

\theoremstyle{definition}
\newtheorem{theorem}{Theorem}[section]
\newtheorem*{theorem*}{Theorem}
\newtheorem{corollary}[theorem]{Corollary}
\newtheorem{lemma}[theorem]{Lemma}
\newtheorem{proposition}[theorem]{Proposition}

\newtheorem{definition}[theorem]{Definition}

\newtheorem{example}[theorem]{Example}

\newtheorem{example*}[theorem]{Example*}
\newtheorem{examples*}[theorem]{Examples*}

\newtheorem{remark*}[theorem]{Remark*}

\newtheorem{convention}[theorem]{Convention}

\if\lecturenotes0

\newtheorem{exer}[theorem]{Exercise}
\newtheorem{exeropt}[theorem]{Exercise}
\newtheorem{exer*}[theorem]{Exercise*}

\fi

\if\lecturenotes1

\newtheorem{exer*}[theorem]{Exercise*}

\newtheoremstyle{exercise}{3pt}{3pt}{\color{red}}{}{\bf}{}{.5em}{}
\theoremstyle{exercise}

\fi

\newcommand{\TODO}[1]{\marginpar{\scriptsize\bB \textbf{TODO:} #1\e}}

\newcommand{\TODOa}[1]{\marginpar{\scriptsize\bM \textbf{TODO:} #1\e}}
\newcommand{\TODOb}[1]{\marginpar{\scriptsize\bB \textbf{TODO:} #1\e}}

\newcommand{\COMMa}[1]{\marginpar{\scriptsize\bM \textbf{COMM:} #1\e}}
\newcommand{\COMMb}[1]{\marginpar{\scriptsize\bB \textbf{COMM:} #1\e}}

\newcommand{\CHECK}[1]{\marginpar{\scriptsize\bR \textbf{CHECK:} #1\e}}

\usepackage{color}
\def\bR{\begin{color}{red}} 
\def\bB{\begin{color}{blue}}
\def\bM{\begin{color}{magenta}}
\def\bC{\begin{color}{cyan}}
\def\bW{\begin{color}{white}}
\def\bBl{\begin{color}{black}} 
\def\bG{\begin{color}{green}}
\def\bY{\begin{color}{yellow}}
\def\e{\end{color}\xspace}
\newcommand{\bit}{\begin{itemize}}
\newcommand{\eit}{\end{itemize}\par\noindent}
\newcommand{\ben}{\begin{enumerate}}
\newcommand{\een}{\end{enumerate}\par\noindent}
\newcommand{\beq}{\begin{equation}}
\newcommand{\eeq}{\end{equation}\par\noindent}
\newcommand{\beqa}{\begin{eqnarray*}}
\newcommand{\eeqa}{\end{eqnarray*}\par\noindent}
\newcommand{\beqn}{\begin{eqnarray}}
\newcommand{\eeqn}{\end{eqnarray}\par\noindent}

\if\lecturenotes1

\renewcommand{\TODO}[1]{}
\renewcommand{\TODOa}[1]{}
\renewcommand{\TODOb}[1]{}
\renewcommand{\COMMa}[1]{}
\renewcommand{\COMMb}[1]{}
\renewcommand{\CHECK}[1]{}

\def\bR{\begin{color}{black}} 
\def\bB{\begin{color}{black}}
\def\bM{\begin{color}{black}}
\def\bC{\begin{color}{black}}
\def\bW{\begin{color}{black}}
\def\bG{\begin{color}{black}}
\def\bY{\begin{color}{black}}

\fi

\usepackage{tikzfig}                    
 
\title{Categorical Quantum Mechanics II:\\  Classical-Quantum Interaction}                                                                                                                                                                      
\author[1]{Bob Coecke}
\author[2]{Aleks Kissinger}

\date{}

\affil[1]{Department of Computer Science, Oxford. {\tt coecke@cs.ox.ac.uk}}            
\affil[2]{iCIS, Radboud University, Nijmegen. {\tt aleks@cs.ru.nl}}

\begin{document}       
\maketitle

\begin{abstract}
This is the second part of a three-part overview, in which we derive the category-theoretic backbone of quantum theory from a process ontology, treating quantum theory as a theory of systems, processes and their interactions.  In this part we focus on classical-quantum interaction.  
  
Classical and quantum systems are treated as distinct types, of which the respective behavioural properties are  specified in terms of processes and their compositions.  In particular, classicality is witnessed by `spiders' which fuse together whenever they connect. We define mixedness and show that pure processes are extremal in the space of all processes, and we define entanglement and show that quantum theory indeed exhibits entanglement. We discuss the classification of tripartite qubit entanglement and show that both the GHZ-state and the W-state  come from spider-like families of processes, which differ only in how they behave when they are connected by two or more wires. We define measurements and provide fully-comprehensive descriptions of several quantum protocols involving classical data flow.   Finally, we give a notion of `genuine quantumness', from which special processes called `phase spiders' arise, and get a first glimpse of quantum non-locality.         
\end{abstract}    


\section{Introduction}  

This is the second part, the first part being \cite{CKpaperI}, of a three-part   overview  on categorical quantum mechanics (CQM), an area of applied category-theory that over the past twelve years has become increasingly prominent within physics, mathematics and computer science, and even has spin-offs in other areas such as computational linguistics. 

Probably the most appealing feature of CQM is the use of diagrams, which are related to the usual Hilbert space model via symmetric monoidal categories and structures therein.  However, we have written this overview in such a way that no prior knowledge on category theory is required. In fact,  it can be seen as a first encounter with the relevant parts of category theory.   

The focus of this part is the role of classicality in quantum theory.  Our stance is somewhat different than the usual account on classicality, where one typically starts from a classical theory and \em quantizes \em it.  Our starting point is instead the process theory of quantum processes as defined in  \cite{CKpaperI}. We then specify what it means for systems and processes to be classical. The motivation for doing so is the fact that classical states admit a property which quantum systems fail to admit, namely that they can be \em broadcast \em \cite{Nobroadcast, CKpaperI}, or, when restricting to pure states, that they can be \em copied \em \cite{Dieks, WZ}.  Therefore, one can characterise classicality by providing classical systems with a witness for this behaviour.  

Once one has this structure at hand, it is easy to define many incarnations of classicality, namely, \em measurement\em, \em decoherence\em, \em mixing\em, \em disentanglement\em, and \em classical control\em. 
Also, once we have identified classicality, it will be possible to identify what `genuine quantumness' actually means.  This will lead to the notion of \em phases\em, and already provides us with a glimpse into \em quantum non-locality\em.   

This three-part overview has a big brother, namely, a forthcoming textbook \cite{CKbook}  by the same authors. While this three-part overview amounts to some 100 pages, the book is about 10 times that size, and provides a more comprehensive introduction suitable for students and researchers in a wide variety of disciplines and levels of experience.

\section{What came before}    

In the first part \cite{CKpaperI} we saw how boxes, a.k.a.~\em processes\em, and wires, a.k.a.~\em systems \em or \em types\em, together make up \em diagrams\em:
\ctikzfig{compound-process}
A \em process theory \em  consists of a collection of systems, a collection of processes with inputs and outputs taken from the systems, and a  means of `wiring processes together'.  So in particular, process theories are `closed under wiring processes together'.  Processes with no inputs are called \em states\em, with no outputs are called \em effects\em, and with no inputs or outputs are called \em numbers\em.  It is then natural to interpret the number arising from the composition of a state and an effect as the \textit{probability} that, given the system is in that state, the effect happens: 
\ctikzfig{state_test_paper}
We called this the  \textit{generalised Born rule}.  

\em String diagrams \em are diagrams where we additionally allow inputs to be connected to inputs and outputs  to be connected to outputs:   
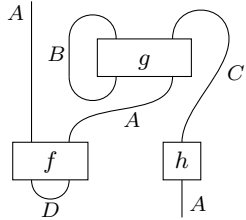
\[
\input{./figures/compound-process-capscups.tikz} 
\]
Two very special string diagrams are \em cups \em and \em caps\em:  
\[
\begin{tikzpicture}
	\begin{pgfonlayer}{nodelayer}
		\node [style=none] (0) at (-1, 0.5) {};
		\node [style=none] (1) at (1, 0.5) {}; 
	\end{pgfonlayer}
	\begin{pgfonlayer}{edgelayer}
		\draw [in=-90, out=-90, looseness=1.75] (0.center) to (1.center); 
	\end{pgfonlayer}
\end{tikzpicture}\qquad\qquad\qquad\qquad\qquad\quad\begin{tikzpicture}
	\begin{pgfonlayer}{nodelayer}
		\node [style=none] (0) at (-1, -0.5) {};
		\node [style=none] (1) at (1, -0.5) {};
	\end{pgfonlayer}
	\begin{pgfonlayer}{edgelayer}
		\draw [in=90, out=90, looseness=1.75] (0.center) to (1.center);
	\end{pgfonlayer}
\end{tikzpicture}  
\]
for which we have the following \em yanking \em equation:
\beq\label{eq:yankforccorrcamp-pre}
\input{./figures/classyanking-CQMII-pre.tikz}
\eeq
The string diagram:
\[
D\ := \ \ \input{./figures/circle.tikz}  
\]
is called \em dimension\em.  Though the numbers in our process theory most definitely don't need to be a field, it is convenient to at least assume we have a special number which cancels out the dimension. That is, for every system, there exists a number, suggestively called `$\oneoverD$', such that:
\[ 
\oneoverD\ \input{./figures/circle.tikz}  \ =\ \emptydiag 
\]
where in the RHS we have the \em empty diagram\em.

Using cups and caps, we define the \em transpose \em of a process to be:  
\[
\begin{tikzpicture}
	\begin{pgfonlayer}{nodelayer}
		\node [style=maptrans] (0) at (0, 0) {$f$};
		\node [style=none] (1) at (0, -1.25) {};
		\node [style=none] (2) at (0, 1.25) {};
	\end{pgfonlayer}
	\begin{pgfonlayer}{edgelayer}
		\draw [style=swap] (0) to (1.center); 
		\draw [style=swap] (2.center) to (0);
	\end{pgfonlayer}
\end{tikzpicture}\ \  :=\ \ \input{./figures/transmapyes.tikz}
\]
In addition to rotating processes 180${}^\circ$, they also can be reflected vertically:
\[
\input{./figures/smapAB.tikz} \ \ \mapsto \ \ \input{./figures/smapdagAB.tikz}  
\]
and we refer to the reflected process as the \em adjoint\em. Combining the transpose and the adjoint we obtain the conjugate:
\ctikzfig{smapALLFOURnew}

A process $U$ is an \em isometry \em if we have the first of these two equations:     
\beq\label{eq:unitary}
\input{./figures/unitary.tikz} 
\eeq
and it is \em unitary \em if we have both.

We defined \em quantum types \em to be doubled wires:
\ctikzfig{doubled_wire}
and \em discarding \em as the effect:
\[
\discard\ \, := \ \ \input{./figures/trace_def.tikz}  
\]
Then, \em quantum processes \em are diagrams of the form:
\beq\label{eq:quantummapx_CQMII}
 \dmap{\Phi} \ :=\ \ \input{./figures/quantummapx_CQMII.tikz}
 \eeq
that are moreover  \em causal\em:
\[ 
\begin{tikzpicture}
	\begin{pgfonlayer}{nodelayer}
		\node [style=none] (0) at (0, -1.25) {};
		\node [style=dmap] (1) at (0, 0) {${\Phi}$};
		\node [style=upground] (2) at (0, 1.5) {};
	\end{pgfonlayer}
	\begin{pgfonlayer}{edgelayer}
		\draw [style=boldedge] (1) to (0);
		\draw [style=boldedge] (1) to (2); 
	\end{pgfonlayer}
\end{tikzpicture}\ \ = \ \, \discard
\] 
\em Pure \em quantum processes are those of the form:    
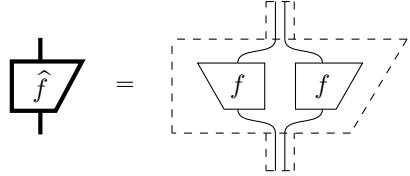
\[
 \input{./figures/double_process_1_to_1z.tikz}
\]
i.e.~they don't involve discarding.

Among other things, doubling guarantees that all numbers are `positive', i.e.~arising from the product of a number with its conjugate:
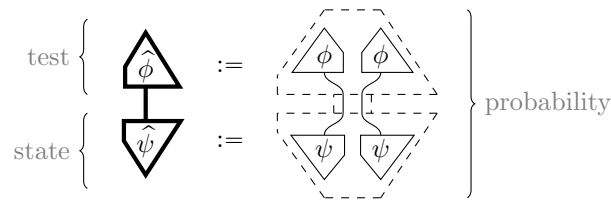
\ctikzfig{doubled_state_test}    

\section{Classical systems and processes}  

\begin{quote}
\em When spider webs unite, they can tie up a lion.\par \em \hfill    --- Ethiopian proverb.                                                                             
\end{quote}

\noindent
We just saw how doubling systems to form quantum systems lets us define a discarding process, which plays a key role both in the construction of impure quantum processes and the formulation of causality.

We will now see that doubling not only makes room for new processes, but new systems as well, namely the classical ones. Over the next section, we will see how it is useful to formulate classical systems as single wires, hence adopting the paradigm:  
 \[
{\mbox{classical} \over\mbox{quantum}} = {\mbox{single wire} \over\mbox{double wires}}
\]
That is, we obtain two distinct types:
\[  
\left( \ \ \textrm{quantum} \ \ := \ \ \didwire \ \  \right)
\quad \neq \quad
\left( \ \ \textrm{classical} \ \ := \ \ \idwire \ \ \right)
\]

We will justify this convention by showing that one can now define processes which \textit{measure} a quantum system to produce classical data, and conversely \textit{encode} classical data within a quantum system:
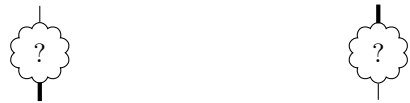
\ctikzfig{CQMII-1}
both of which are key to understanding quantum features. We will build these processes using...

\subsection{Spiders}

How can we make our single, classical wires act in a classical way? Unlike quantum systems, which can only be used once, classical data can be copied and deleted. We can witness this fact by introducing special processes on a classical wire for copying and deleting:
\[
\textit{copy} \ :=\ \input{./figures/copy.tikz}
\qquad\qquad\qquad\quad
\textit{delete} \ :=\ \begin{tikzpicture}
	\begin{pgfonlayer}{nodelayer}
		\node [style=white dot] (0) at (0, 0.5) {};
		\node [style=none] (1) at (0, -0.5) {};
	\end{pgfonlayer}
	\begin{pgfonlayer}{edgelayer}
		\draw [style=swap] (0) to (1.center); 
	\end{pgfonlayer}
\end{tikzpicture} 
\]

In fact, these two processes are part of a whole family of `generalised wires':

\begin{definition}\label{def:spider}
A \em spider \em is a diagram of the following form:
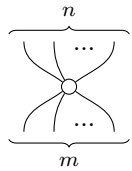
\ctikzfig{spidernm}
and the behaviour of a family of spiders (i.e.~one spider for each $n$ and each $m$) is subject to the following rules:
\bit
\item Spider-fusion:
\beq\label{eq:bastards-compose-bis}
\input{./figures/bastards-compose-bis.tikz} 
\eeq
\item Invariance under `leg-swapping':  
\[ 
\input{./figures/spider-flip.tikz}  \ \  =\ \  \input{./figures/spider-flip2.tikz}
\] 
\item Invariance under conjugation: 
\[ 
\input{./figures/spider-hor.tikz} \ \  =\ \  \input{./figures/spider.tikz} 
\] 
\item A single wire, cups and caps are spiders:
\[ 
     \begin{tikzpicture}
	\begin{pgfonlayer}{nodelayer}
		\node [style=none] (0) at (0, 1.25) {};
		\node [style=none] (1) at (0, -1.25) {};
		\node [style=white dot] (2) at (0, 0) {};
	\end{pgfonlayer}
	\begin{pgfonlayer}{edgelayer}
		\draw [style=swap] (0.center) to (2);
		\draw [style=swap] (2) to (1.center);
	\end{pgfonlayer}
\end{tikzpicture} \ \ =\ \ 
   \begin{tikzpicture}
	\path [use as bounding box] (-0.25,-1) rectangle (0.25,1);
	\begin{pgfonlayer}{nodelayer}
		\node [style=none] (0) at (0, 1.25) {};
		\node [style=none] (1) at (0, -1.25) {};
	\end{pgfonlayer}
	\begin{pgfonlayer}{edgelayer}
		\draw (1.center) to (0.center);
	\end{pgfonlayer}
\end{tikzpicture}
   \qquad\quad\ \
   \input{./figures/correlate.tikz} \ \ =\ \ \input{./figures/correlatecopy.tikz}  
   \qquad\quad\ \
   \input{./figures/compare.tikz} \ \ =\ \ \input{./figures/comparecopy.tikz}    
\] 
\eit
\end{definition}

Spiders can indeed be seen as `generalised wires'. Whereas normal wires only connect a single input to a single output, spiders connect many inputs to many outputs. Then, spider-fusion above witnesses the fact that, just like connecting two ordinary wires yields again an ordinary wire, connecting two generalised wires yields again a generalised wire. In other words, this sort of `connection' is transitive: if a wire $w$ is connected to $w'$ via a spider and $w'$ is connected to $w''$, then $w$ is connected to $w''$.  A consequence of this  is that any \underline{connected} string diagram consisting only of spiders is equal to a \underline{single} spider: 
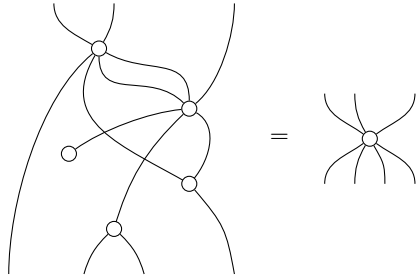
\[ 
\input{./figures/spidercompfuse.tikz}
\] 
Hence, such a diagram is uniquely determined by its number of inputs and outputs, and we can establish equality simply by counting them.  


If we specialise the spider rules to copying and deleting, we get exactly the behaviours we expect. For instance, copying then deleting one copy does nothing:
\beq\label{eq:copydelete2}
\input{./figures/copydelete2.tikz}
\eeq
copies can be swapped:
\beq\label{eq:copyswapa}
\input{./figures/copyswapa.tikz}   
\eeq
and the following two ways of producing three copies are equal:  
\beq\label{eq:copycopya}
\input{./figures/copycopya.tikz}
\eeq

One particularly interesting instance of spider-fusion is:
\[ 
\input{./figures/classyanking-CQMII.tikz}
\] 
which, when replacing these spiders with cups, caps and plain wires, is equation (\ref{eq:yankforccorrcamp-pre}). In other words  the yanking equation for string diagrams arises as an instance of spider-fusion.

With all this talk about copying, it is natural to ask what exactly is being copied here. The answer is: classical values. But rather than giving a separate notion of a classical value, then showing it is copied by a spider, copiability is actually what \underline{defines} classical values:

\begin{definition}
A \textit{classical value} is a state $i$ of a classical system such that:
\beq\label{eq:copyableonly}
\input{./figures/copyableonly.tikz} \ \ =\  \raisebox{-1mm}{\point{i}} \raisebox{-1mm}{\point{i}}  
\eeq
\end{definition}

\noindent Note that our notation for classical values indicates that classical values are self-conjugate, since conjugation plays no role in classical systems.

Next we will see how classical values can travel to and from the quantum world, via...

\subsection{Classical-quantum processes}  

One of the basic ingredients for quantum computing is the ability to encode classical data within a quantum system. We can do this simply by doubling a classical value to obtain the associated quantum state:
\[ 
\pointmap{i} \ \ \mapsto\ \ \dpoint{i}\ :=\ \ \input{./figures/dval.tikz} 
\]

The process which performs this encoding is just the copy-spider, where the two output legs are interpreted as a single quantum system:
\ctikzfig{encode-derive}

Hence, spiders let us construct a bridge between classical and quantum systems. The most important examples are \textit{encoding} and its adjoint \textit{measuring}, which are defined respectively as:
\[ 
\input{./figures/encode.tikz}
\qquad\qquad\qquad  \qquad  \quad
\input{./figures/measure.tikz}
\] 
In Section~\ref{sec:ClassProc}, we will see that measuring actually sends a quantum state to a probability distribution computed according to the Born rule. But for now, the following should adequately justify the term `measuring':

\begin{proposition}
Measuring undoes encoding:
\[ 
  \input{./figures/encodemeasure.tikz}
\] 
\end{proposition}
\begin{proof}
Unfolding doubled wires we have:
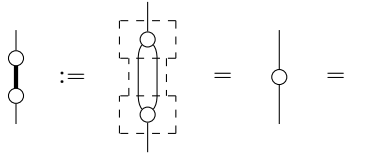
\[
\input{./figures/measureencode-CQMII.tikz}  
\]
where we used spider-fusion.
\end{proof}

Measuring and encoding now allow us to build processes with both classical and quantum inputs/outputs:
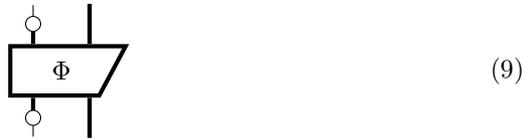
\begin{equation}\label{eq:cq-form} 
   \input{./figures/classquantmap1.tikz}
\end{equation}
By (\ref{eq:quantummapx_CQMII}) we know that any quantum process can be expressed in terms of a pure quantum process and discarding. Hence:
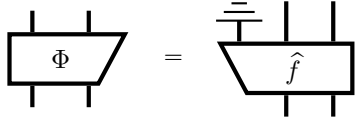
\ctikzfig{c2-purify}
So, \eqref{eq:cq-form} can be explicitly unfolded as follows:
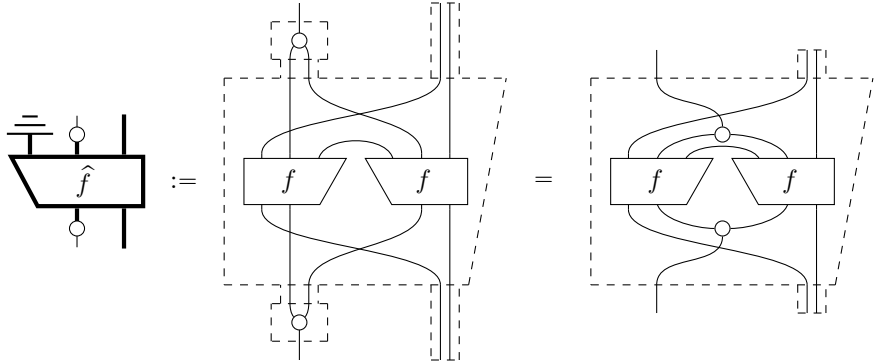
\[
\input{./figures/classquantmap2.tikz} \ \  :=\ \  \input{./figures/double_classquant.tikz}   
\]
Furthermore, we can extend Proposition 4.2 of \cite{CKpaperI} to diagrams which contain not only pure processes and discarding, but also those involving classical-quantum interactions:

\begin{proposition}\label{prop:cq-boxclosedness}
Any string diagram consisting of processes of the form (\ref{eq:cq-form}) is again of that form. Hence they form a process theory.
\end{proposition}
\begin{proof}
Firstly, by re-ordering some inputs/outputs, we can bring all of the measure/encode processes together:  
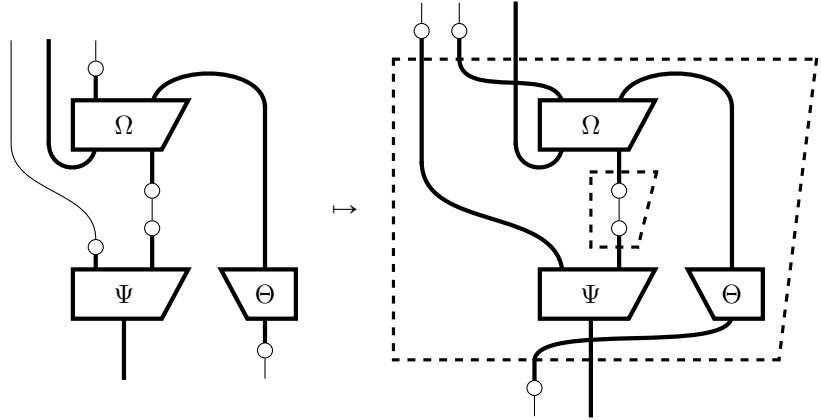
\[
  \input{./figures/arbitrary_w_me1_copy1.tikz}\quad \mapsto\quad \input{./figures/arbitrary_w_me2_copy.tikz}  
\]
So it remains to be shown that the process in the large dashed box is of the form (\ref{eq:quantummapx_CQMII}).  This amounts to showing all of the boxes it contains are of the form (\ref{eq:quantummapx_CQMII}). The only one possibility  
which doesn't follow immediately by assumption is the process in the small dashed box. In that case, we have:
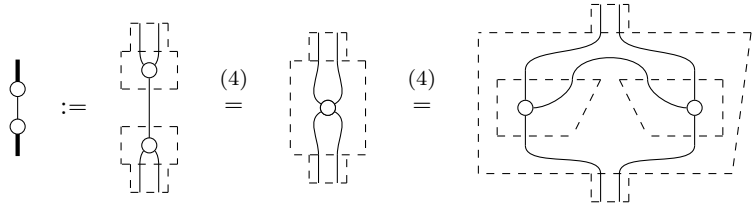
\[
\input{./figures/c2-decoherenceiscqmap.tikz}
\]
so this process is also of the form (\ref{eq:quantummapx_CQMII}). 
\end{proof}

So how does causality extend to these hybrid boxes?  The conceptual underpinning of causality for quantum processes is that if we apply  a process to some inputs, but then discard all of its outputs, its performance  has gone to waste, and we could as well  have simply discarded its inputs.  This principle evidently extends to hybrid boxes when we also delete its classical outputs, which yields the following generalisation of our definition of a quantum process, now also allowing for classical inputs and outputs:  

\begin{definition}
A \em quantum process \em is a process of the form:
   \ctikzfig{classquantmap1}
which is causal:
\[ 
  \input{./figures/classquantmapcausality.tikz}
\] 
\end{definition}

Note that we still simply call these `quantum processes'. The prior definition now arises as a special case where a process has no classical inputs/outputs. As it should be understood that a quantum process could have zero or more classical and quantum outputs, causality means deleting/discarding \underline{all} of the outputs amounts to deleting/discarding \underline{all} of the inputs.
Thanks to the underlying spider-definitions, these two notions are directly related:


\begin{proposition}
Encoding then discarding is the same as deleting, and measuring then deleting is the same as discarding:
\beq\label{eq:grounddecompa-CQMII}
\input{./figures/grounddecompb-CQMII.tikz}
\qquad\qquad\qquad \quad
\input{./figures/grounddecompa-CQMII.tikz}  
\eeq
\end{proposition}
\begin{proof}
We will show the measuring/deleting case, as the other equation is proven similar. Unfolding doubled wires we have:
\[
\input{./figures/measuredelete.tikz}
\]
\end{proof}

The equations~\eqref{eq:grounddecompa-CQMII} are exactly the causality equations for encoding and measuring, hence:

\begin{corollary}
  Encoding and measuring are \textit{quantum processes}.
\end{corollary}

In Section 4.2 of \cite{CKpaperI} we showed that if the \em reduced process \em (i.e.~what remains if we discard some quantum output) of a process  is pure,  then  it $\otimes$-separates. In order to prove this, we had to assume the following `diagrammatically natural' condition:
\beq\label{EQ:daggeraxiom}
\left( \exists \psi, \phi: \ \ \map{f} \ =\ \kpointketbra{\psi}{\phi} \right)
  \ \ \Longleftrightarrow\ \
 \left( \exists \psi', \phi':  \ \   \begin{tikzpicture}
	\begin{pgfonlayer}{nodelayer}
		\node [style=mapdag] (0) at (0, 1) {$f$};
		\node [style=map] (1) at (0, -1) {$f$};
		\node [style=none] (2) at (0, -2.25) {};
		\node [style=none] (3) at (0, 2.25) {};  
	\end{pgfonlayer}
	\begin{pgfonlayer}{edgelayer}
		\draw (1) to (2.center);
		\draw (3.center) to (0); 
		\draw (0) to (1);
	\end{pgfonlayer}
\end{tikzpicture}\ \ =\ \kpointketbra{\psi'}{\phi'}\right)  
\eeq
Not surprisingly, the same holds when  we delete a classical output:

\begin{proposition}\label{prop:discard-mix-pure-proc-class}  
Assuming (\ref{EQ:daggeraxiom}), if we have:
  \begin{equation}\label{eq:rho-disc-pure-proc-class}
\input{./figures/bekan7class.tikz}\  \ =\ \dmap{\widehat f}  
  \end{equation}
then the quantum process separates as follows:
 \beq\label{eq:bekan8class}
\input{./figures/bekan8class.tikz}\  \ =\ \point{p}\ \dmap{\widehat f}    
  \eeq
\end{proposition}
\begin{proof}
By equation (\ref{eq:grounddecompa-CQMII}), equation (\ref{eq:rho-disc-pure-proc-class}) is equivalent to:  
\[
\input{./figures/bekan7.tikz}\  \ =\ \dmap{\widehat f}  
\]
so by Proposition  4.4 of \cite{CKpaperI} we have:
\[
\input{./figures/bekan8class.tikz}\ \ = \ \ \input{./figures/bekan42withf.tikz}
\]
Setting:
\[
\point{p}\  := \ \ \raisebox{1.5mm}{\input{./figures/bekan42.tikz}}  
\]
we obtain the form (\ref{eq:bekan8class}).
\end{proof}

We didn't specify in the statement of Proposition \ref{prop:discard-mix-pure-proc-class} what $p$ is.  It is a `classical state', which is a special case of...

\subsection{Classical processes}\label{sec:ClassProc}

Classical processes are a special case of quantum processes:

\begin{definition}
A \em classical process \em is a process of the form:  
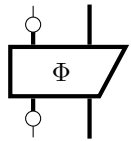
\begin{equation}\label{eq:classical-map}
\map{f}\ \, :=\ \ \input{./figures/class-map.tikz}
\end{equation}
which is causal:    
\beq\label{eq:classcaus}
\begin{tikzpicture}
	\begin{pgfonlayer}{nodelayer}
		\node [style=map] (0) at (0, 0) {$f$};
		\node [style=white dot] (1) at (0, 1.25) {}; 
		\node [style=none] (2) at (0, -1.25) {};
	\end{pgfonlayer}
	\begin{pgfonlayer}{edgelayer}
		\draw (2.center) to (0);
		\draw (0) to (1);
	\end{pgfonlayer}
\end{tikzpicture}\ \ =\ \ \begin{tikzpicture}
	\begin{pgfonlayer}{nodelayer}
		\node [style=white dot] (0) at (0, 0.5) {};
		\node [style=none] (1) at (0, -0.5) {};
	\end{pgfonlayer}
	\begin{pgfonlayer}{edgelayer}
		\draw [style=swap] (0) to (1.center); 
	\end{pgfonlayer}
\end{tikzpicture}  
\eeq
\end{definition}

So we now have three kinds of processes:
\ctikzfig{qc-map-generalise-CQMII}
namely those which are totally classical, totally quantum, or some combination thereof. Classical processes are also called \em stochastic processes \em and classical states: 
\beq\label{eq:measure-rho}
\point{p}\  := \ \ \begin{tikzpicture}
	\begin{pgfonlayer}{nodelayer}
		\node [style=white dot] (0) at (0, 0.5) {};
		\node [style=dkpoint] (1) at (0, -0.75) {$\boldsymbol\rho$};
		\node [style=none] (2) at (0, 1.25) {};
	\end{pgfonlayer}
	\begin{pgfonlayer}{edgelayer}
		\draw [style=swap] (2.center) to (0);
		\draw [style=boldedge] (0) to (1);
	\end{pgfonlayer}
\end{tikzpicture}
\eeq
are also called \em probability distributions\em.  The state:  
  \beq\label{eq:uniform-probability-distribution}
  \oneoverD \ \bigunit{white dot} 
  \eeq
  is called the \textit{uniform probability distribution}, and
  \[
\oneoverD \ \input{./figures/correlate.tikz}  
\]
  is called \textit{perfect classical correlations}. 

We already saw two other examples of classical processes, namely copying and deleting,  in fact, they are examples of a special kind of classical processes:  

\begin{definition}
A classical process is \em deterministic \em if we have:  
\begin{equation}\label{eq:function-map-char}
    \input{./figures/classical_map1.tikz}
\end{equation}
\end{definition}

So deterministic classical processes are processes which are such that, applying them, and then copying, is the same as first copying and then applying the process to each of the copies.  

We actually haven't shown yet that copying and deleting are classical processes, and rather than doing that directly, we prove something more general:

\begin{proposition}
Any process satisfying equation (\ref{eq:function-map-char}) as well as the causality equation (\ref{eq:classcaus}), and that is also self-conjugate:
\beq\label{eq:selfconjdeterministic}
\mapconj{f}\ \, = \ \, \map{f}
\eeq
is always a classical process, and hence, a deterministic classical process.
\end{proposition}
\begin{proof}
By spider-(un)fusion we have:
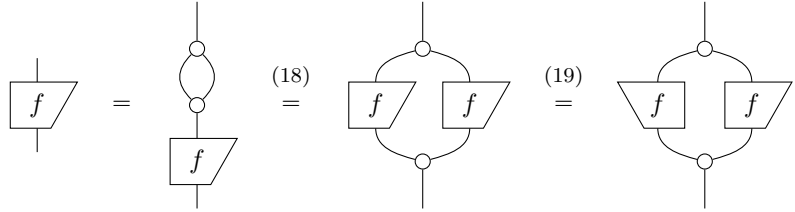
\ctikzfig{CQMII-deterministicproof}
so we indeed obtain the required form (\ref{eq:classical-map}):    
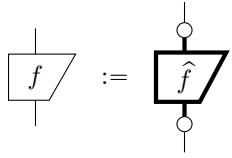
\[
\map{f}\ \, :=\ \ \input{./figures/class-map-bis.tikz}
\]
\end{proof}

Using spider fusion,  it is easy to see that copying and deleting themselves satisfy equations (\ref{eq:function-map-char}) and that they are causal:
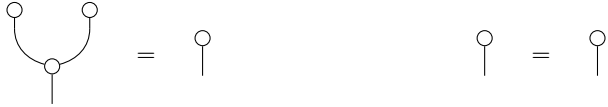
\ctikzfig{copy-causal-CQMII}
As spiders they are also self-conjugate.   

A special case of deterministic classical processes are classical values, which we saw before. Since their input system is trivial (i.e.~`no wire') equation~\eqref{eq:function-map-char} reduces to the copying equation:
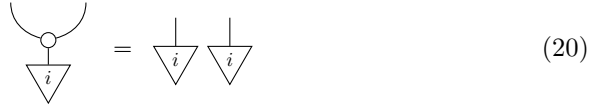
\begin{equation}\label{eq:copy-equation}
  \input{./figures/copyableonly.tikz} \ \ =\  \raisebox{-1mm}{\point{i}} \raisebox{-1mm}{\point{i}}
\end{equation}
These are therefore the deterministic states of a classical system, whereas it is useful to think of arbitrary states:
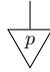
\[
\point{p}
\]
as \textit{probability distributions} over these classical values. Indeed to each such state, we can associated a list of numbers:
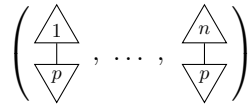
\[
\left(\innerprodmap{1}{p}\ ,\ \ldots\ ,\ \innerprodmap{n}{p}\right)  
\]
which give the probability that our system is in each of the respective deterministic states.

When we look at classical states arising from measuring a quantum state as in \eqref{eq:measure-rho}, this list of probabilities is:
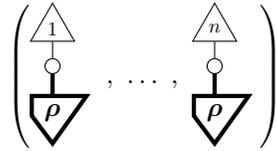
\[
\left( \raisebox{-1.1mm}{\begin{tikzpicture}
	\begin{pgfonlayer}{nodelayer}
		\node [style=white dot] (0) at (0, 0.5) {};
		\node [style=dkpoint] (1) at (0, -0.75) {$\boldsymbol\rho$};
		\node [style=copoint] (2) at (0, 1.5) {$1$};
	\end{pgfonlayer}
	\begin{pgfonlayer}{edgelayer}
		\draw [style=swap] (2) to (0);
		\draw [style=boldedge] (0) to (1);
	\end{pgfonlayer}
\end{tikzpicture}}\ ,\ \ldots\ ,\ \raisebox{-1.1mm}{\begin{tikzpicture}
	\begin{pgfonlayer}{nodelayer}
		\node [style=white dot] (0) at (0, 0.5) {};
		\node [style=dkpoint] (1) at (0, -0.75) {$\boldsymbol\rho$};
		\node [style=copoint] (2) at (0, 1.5) {$n$};
	\end{pgfonlayer}
	\begin{pgfonlayer}{edgelayer}
		\draw [style=swap] (2) to (0);
		\draw [style=boldedge] (0) to (1);
	\end{pgfonlayer}
\end{tikzpicture}} \right)
\]
By the copying equation \eqref{eq:copy-equation}, this becomes:
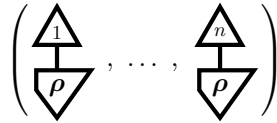
\[
\left(
\kpointbraket[style1=dkpoint,style2=dcopoint,estyle=boldedge]{\bm\rho}{1}
\ ,\ \ldots\ ,\ 
\kpointbraket[style1=dkpoint,style2=dcopoint,estyle=boldedge]{\bm\rho}{n}
\right)
\]
Hence, the resulting classical state is a probability distribution consisting of all of the Born-rule probabilities for $\bm\rho$ being in each of the states:
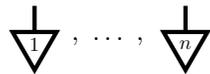
\[ 
\dpoint{1}\ ,\ \ldots\ , \ \dpoint{n}   
\]
This probability distribution contains some info about $\bm\rho$, but it also throws a lot away. Hence, passing through a classical system always yields a destructive process, known as... 

\subsection{Decoherence}

This is the quantum process:
\ctikzfig{decoherence1-CQMII}
that appeared in the proof of Proposition \ref{prop:cq-boxclosedness}. While it is a quantum process without classical inputs or outputs, it has a close connection with classicality. More specifically, the middle bit is classical:
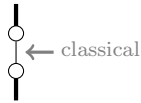
\ctikzfig{decoherence1-CQMIIcopy}
so it forces everything that passes through to be classical for a while.  More generally, when we compose decoherence with the inputs and outputs of any quantum process, the `core' of the process becomes classical:   
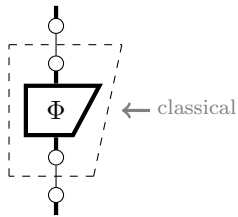
\ctikzfig{CQMII-2}
Hence, decoherence eliminates any `quantumness' inherent in a quantum process. In Section \ref{sec:quantumness} we will make this statement precise, after we define what `quantumness' actually is.  For now, it suffices to say decoherence acts non-trivially on quantum states, and hence is not equal to the identity. In fact, we can say something  stronger:  

\begin{proposition}\label{prop:decoherence}
Decoherence is not pure.
\end{proposition}
\begin{proof}
Suppose decoherence is pure. Then, for some $f$:
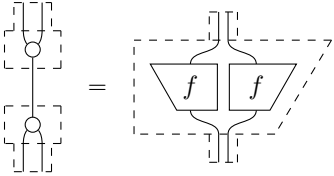
\[
\input{./figures/decoh-unfold.tikz}\ \ =\ \ \input{./figures/double_process_1_to_1x.tikz} 
\]
But then, using  spider-fusion, we have:
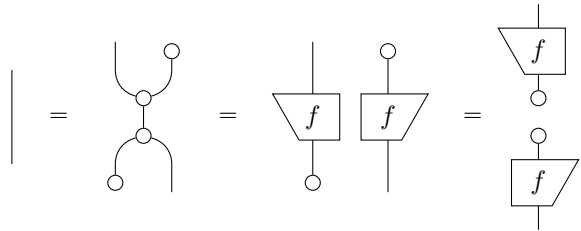
\ctikzfig{decoherenceqmap2new}
i.e.~the identity is disconnected, so decoherence cannot be pure.     
\end{proof}

Since the identity  is pure, decoherence cannot be the identity.  

\subsection{Quantum and bastard spiders}\label{sec:quantum-and-bastard}

Decoherence is an example of a quantum process that is made up of spiders.  However, simply by doubling  `classical' spider, we obtain a \em quantum spider\em:
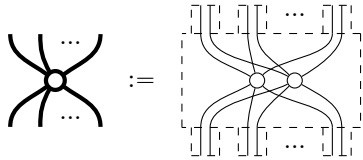
\[ 
\input{./figures/zon2.tikz}\ \ := \ \  \input{./figures/qspider2s.tikz} 
\]
Important examples are the \em Bell-state \em and the \em GHZ-state\em:  
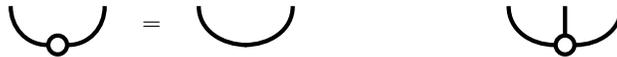
\[
\input{./figures/qspidera.tikz} \qquad\qquad\qquad\qquad\input{./figures/qspiderb.tikz}    
\]

Now, neither measuring, encoding, nor decoherence are of the form of a quantum spider (cf.~Proposition \ref{prop:decoherence}). So, what are they? They are examples of the most general kind of spiders that arise when we use measure and encode to connect any classical spider to any quantum spider:   
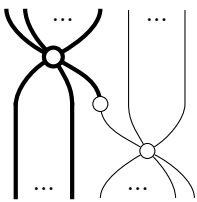
\ctikzfig{qspiderh}
When unfolding the doubled spider we see that we can fuse all the spiders together into a single spider:
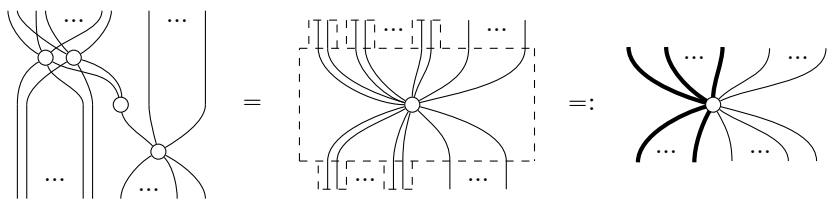
\ctikzfig{CQMII-3}
where some of the legs are quantum and some are classical.  We call these \em bastard spiders\em.  Another example of a bastard spider that we already encountered is discarding:
  \beq\label{eq:discard-is-bastard}
  \input{./figures/discard-is-bastard.tikz}
  \eeq
Bastard spiders have their own dedicated fusion rule:    
  
\begin{proposition}\label{thm:bastards-compose}
Any composition of spiders involving at least one `single-head' spider yields a `single-head' spider, for example:
 \[
\input{./figures/bastards-eat-others-restrict_c.tikz}\qquad\quad\ \input{./figures/bastards-eat-others-restrict_q.tikz}    
\]
\end{proposition}

\subsection{Category-theoretic counterpart}    

In \cite{CKpaperI} we already encountered dagger-symmetric monoidal categories ($\dagger$-SMCs). Usually, spiders are defined as a (co)algebraic structure within a $\dagger$-SMC: they consist of an algebraic part (a  monoid), a co-algebraic part (a comonoid), and some interaction laws between the two:

\begin{definition}
A \textit{monoid} in an SMC consists of an object $A$ and a pair of morphisms:
\[ 
\whitemult \ :\  A \otimes A \to A \qquad\quad\qquad\qquad\whiteunit \ :\ I \to A 
\]
such that \whitemult is \em associative \em and has \em unit \em \whiteunit:
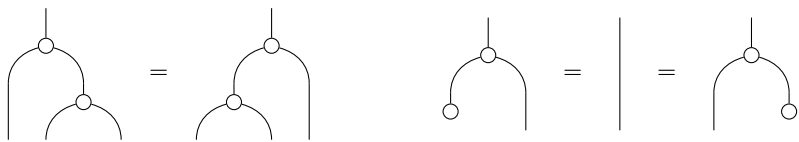
\[ 
\input{./figures/copycopyb.tikz} \qquad\qquad\quad \input{./figures/copydelete2dag.tikz} 
\]
Similarly, a \textit{comonoid} consists of an object $A$ and morphisms:  
\[ 
\whitecomult \  :\ A \to A \otimes A \qquad\qquad\qquad\quad  \whitecounit \   :\ A \to I 
   \]
such that \whitecomult is \textit{coassociative} and has \textit{counit} \whitecounit:
\[ 
\input{./figures/copycopya.tikz} \qquad\qquad\quad \input{./figures/copydelete2.tikz} 
\]
A \textit{Frobenius algebra} $(A, \whitemult, \whiteunit, \whitecomult, \whitecounit)$ consists of a monoid and a comonoid on the same object $A$, which satisfies the \textit{Frobenius equations}:
\ctikzfig{frobenius}
A \textit{dagger-special commutative Frobenius algebra} ($\dagger$-SCFA) in a $\dagger$-SMC is a Frobenius algebra that additionally satisfies:
  \[ 
  \whitecomult \ =\ \left(\whitemult\right)^\dagger \qquad\qquad    
     \input{./figures/copymatch.tikz} \qquad\qquad
     \input{./figures/copyswapb.tikz} 
     \]
\end{definition}

Spiders now arise from $\dagger$-SCFAs as follows:
\[
\input{./figures/spidernm.tikz} \ \ :=\ \  \input{./figures/S-mn-tree.tikz}
\]

\begin{example}\label{ex:hilb-onb}
In the category FHilb whose objects are finite-dimensional Hilbert spaces and whose morphisms are linear maps, families of spiders, or equivalently, $\dagger$-SCFAs, are in 1-to-1 correspondence with orthonormal bases (ONBs). More specifically,  every family of spiders defines a unique ONB as follows:
\[
\left\{  \raisebox{-1mm}{\point{i}} \  \middle|\ \  \input{./figures/copyableonly.tikz} \ \ =\  \raisebox{-1mm}{\point{i}} \raisebox{-1mm}{\point{i}} \ \right\}
\]
That is, the ONB arises as the states that are copied by the copy-spider.  Conversely, given an ONB, a family of spiders arise as follows:
\begin{equation}\label{eq:spider-concrete}
  \input{./figures/spidernm.tikz}\ \  := \ \ \input{./figures/spiderdefNEW.tikz}
\end{equation}
So, in other words, spiders constitute a diagrammatic axiomatisation of ONBs.    
\end{example}

A diagram which includes (chosen) spiders for each type is called a \textit{spider diagram}:
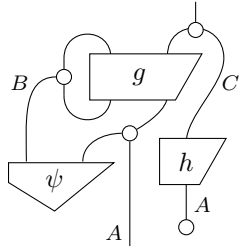
\ctikzfig{compound-process-dots}
This can be seen as a natural progression from string diagrams. For string diagrams, we relaxed the notion of a diagram to allow inputs to be connected to inputs and outputs to outputs. However, we still required that any input/output must be connected to at most one other input/output. Spider diagrams relax this restriction by allowed any number of inputs and outputs to be connected, via a spider.  As with string diagrams, spider diagrams have a category-theoretic counterpart:

\begin{definition}
  A (dagger) \textit{hypergraph category} is a category where each object $A$ has a chosen (dagger) special commutative Frobenius algebra.
\end{definition}

Hence we can extend the table from Section 3.5 in \cite{CKpaperI} with one more column:

\begin{center}
\begin{tabular}{c|c|c|c|}     
\hline
 & \em string diagram \em & \em \textbf{spider diagram} \em \\
  $\cdots$ & \input{./figures/compound-process-capscups-book.tikz} &
\input{./figures/compound-process-dots.tikz}  \\
& i.e.~1 in/out to 1 in/out &
i.e.~$m$ ins/outs to $n$ ins/outs  \\
\hline
$\cdots$ & $\Rightarrow$ compact closed category & $\Rightarrow$ hypergraph category \\    
\hline
\end{tabular}
\end{center}

As we saw in Example~\ref{ex:hilb-onb}, in the category of finite-dimensional Hilbert spaces, fixing a $\dagger$-SCFA for each object is the same as choosing an ONB for each Hilbert space. The resulting category is equivalent to the following category:

\begin{definition}
  $\textrm{Mat}(\mathbb C)$ is the SMC whose objects are natural numbers, and whose morphisms $M : m \to n$ are $n \times m$ matrices, where:
  \begin{itemize}
    \item $\circ$ is multiplication of matrices
    \item $m \otimes n = mn$ on objects, and $M \otimes N$ is the Kronecker product of matrices.
  \end{itemize}
\end{definition}

Every object $n$ in $\textrm{Mat}(\mathbb C)$ has a canonical choice of ONB given by the $n \times 1$ matrices:
\[ \point{1} \ :=\  \left(\begin{matrix}
  1 \\ 0 \\ \vdots \\ 0
\end{matrix}\right)\qquad
\point{2} \ :=\  \left(\begin{matrix}
  0 \\ 1 \\ \vdots \\ 0
\end{matrix}\right)\qquad
\cdots\qquad
\point{n} \ :=\  \left(\begin{matrix}  
  0 \\ 0 \\ \vdots \\ 1
\end{matrix}\right) \]
with which one can define spiders as in equation~\eqref{eq:spider-concrete}. Letting adjoints be the conjugate-transpose of matrices, we have:

\begin{theorem}
  $\textrm{Mat}(\mathbb C)$ is a dagger hypergraph category.
\end{theorem}

Furthermore, as we saw with string diagrams and FHilb, any equation between spider diagrams in $\textrm{Mat}(\mathbb C)$ involving matrices $M, N, \dots$ holds \textit{generically}--i.e.~for \underline{all} matrices $M, N, \ldots$--precisely when the diagrams themselves are equal. In other words:

\begin{theorem}\label{thm:spider-completeness}
  $\textrm{Mat}(\mathbb C)$ is complete for spider diagrams.
\end{theorem}

Now, we introduced quantum and classical systems as two different kinds of things: doubled wires and single wires (which come with a family of spiders).  However, there is a manner in which classical systems and quantum systems, as well as hybrids thereof, can be treated on the same footing, via certain non-commutative Frobenius algebras:

\begin{definition}\label{def:dagger-SSFA}
  A dagger-special \textit{symmetric} Frobenius algebra ($\dagger$-SSFA) in a $\dagger$-SMC is a Frobenius algebra that additionally satisfies:
  \[ 
  \whitecomult \ =\ \left(\whitemult\right)^\dagger \qquad\qquad
     \input{./figures/copymatch.tikz} \qquad\qquad
     \input{./figures/symmetric.tikz} 
     \]
\end{definition}

While it doesn't look much different, the third, \textit{symmetry equation} above is a vast generalisation of commutativity. If we now look at morphisms which factor as:
\begin{equation}\label{eq:cpstar-form}
  \input{./figures/classquantmap1class-paper.tikz} \ \ := \ \  \input{./figures/double_classquantclass-paper.tikz}
\end{equation}
(which is the diagram we used to define classical processes)  and generalise to the measure/encode processes  to those coming from a $\dagger$-SSFA, this  splits into multiple cases. If the measure/encode processes come from a  commutative $\dagger$-SSFA (i.e.~a $\dagger$-SCFA), this is just a classical process as before. But if they come from a non-commutative $\dagger$-SSFA, these represent non-classical processes. For example, taking the $\dagger$-SSFAs to be the following `pants algebra':
\ctikzfig{pants-alg}
then the form given in \eqref{eq:cpstar-form} is equivalent to the condition for being of the form (\ref{eq:quantummapx_CQMII}). 
Furthermore, since $\otimes$-compositions of $\dagger$-SSFAs are again $\dagger$-SSFAs, \eqref{eq:cpstar-form} suffices to capture classical systems, quantum systems, and $\otimes$-compositions thereof. This fact then motivates the following definition, which relies on the dagger-compact closed categories of Definition 3.13 in \cite{CKpaperI}:

\begin{definition}
  For any dagger-compact closed category $\mathcal C$, $\textbf{CP*}[\mathcal C]$ has as objects $\dagger$-SSFAs and as morphisms:
  \[
  f : (A, \whitemult, \ldots) \to (B, \graymult, \ldots)  
  \]
  morphisms $f : A \to B$ from $\mathcal C$, which additionally satisfy \eqref{eq:cpstar-form}.
\end{definition}

\subsection{Reference and further reading}  

The diagrammatic representation of classical data was initiated in \cite{CPav}, in which spider-less but fully diagrammatic descriptions of several protocols were given. The passage to spiders took place in \cite{CPaq}. Much earlier, in \cite{DaviesLewis}, the classical data produced by quantum measurements was also represented in terms of an ONB.  However, no distinction was made between the spaces in which the classical data was represented and in which the quantum systems were described. 

The two vs.~one wire paradigm and an early form of the resulting quantum processes were introduced in \cite{CPaqPav}. Bastard spiders are a refinement of the notation used in \cite{CDKZ, CDKZ2} to represent ONB-measurements. The fact that spiders characterise ONBs is from \cite{CPV}.     

The notion of a Frobenius algebra first appeared in \cite{BrauerNesbitt} but was first presented in its modern, categorical form in \cite{CarboniWalters}. The fact that (special) commutative Frobenius algebras have canonical forms (i.e.~spiders) that `fuse' together comes from a `folk theorem' relating Frobenius algebras to geometrical objects called \textit{cobordisms}. A standard reference is \cite{KockBook}. An explicit proof of the `spider' form for special commutative Frobenius algebras was given using distributive laws in \cite{Lack}. Completeness for spider diagrams was given in \cite{KissingerCompleteness}  and the CP*-construction was introduced in \cite{CPstar}.

\section{Classical-quantum interaction}    

\begin{quote}  
\em Damn it! I knew she was a monster! John! Amy! Listen! Guard your *********.\par \em \hfill    --- David Wong, This Book Is Full of Spiders, 2012.    
\end{quote}            

\noindent  
Classical systems play several roles in quantum theory.  We already saw that the spiders which `witness' classicality allow us to define the measurement and encoding processes which form the interfaces between classical and quantum systems, and how decoherence is a quantum process that invokes classical behaviour.  We'll now discuss some more roles of classicality.

\subsection{Measurement}

Earlier we already encountered one kind of quantum measurement:
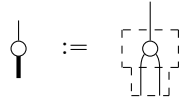
\ctikzfig{measure} 
More general measurements arise as follows:
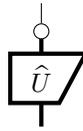
\ctikzfig{dem-ONB}
  where $\widehat U$ is a unitary quantum process.   All of these measurements are called \em demolition \em measurements, since after the measurement the quantum system is gone.  An example of a \em non-demolition \em measurement is this bastard spider:
 \beq\label{eq:CQMII-7} 
 \input{./figures/CQMII-7.tikz}
  \eeq
 or, more generally:
 \beq\label{eq:non-dem-map} 
  \input{./figures/non-dem-map.tikz}    
  \eeq
 If we discard the quantum output of a non-demolition measurement, we obtain the demolition measurement we had before:  
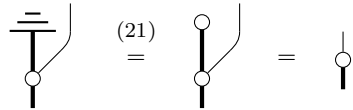
\ctikzfig{nondem-disc-quantum}
On the other hand, if we discard the classical output, we get decoherence:    
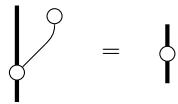
\[ 
  \input{./figures/CQMII-8.tikz}
\] 
So performing a non-demolition measurement, and then forgetting about the outcome yields a non-trivial operation.  

What does such a measurement do to a quantum state?  Just like we used deterministic classical states to pick out probabilities from a probability distribution at the end of Section \ref{sec:ClassProc}, we can use them to see how a measurement alters the state, by plugging them in the measurement-outcome wire:
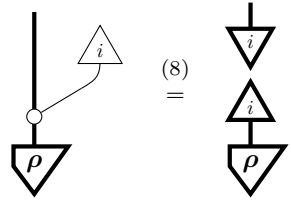
\ctikzfig{CQMII-9}
Hence,  the state undergoes the following radical transition:
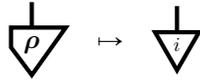
\[
\dkpoint{\boldsymbol\rho}\ \ \mapsto\ \ \raisebox{1mm}{\dpoint{i}}  
\]
usually referred to as \em collapse\em.
 
All of the above measurements are called \em non-degenerate\em, due to the manner they are defined in terms of unitaries and spiders only.  A more general kind of measurement is called a \em von Neumann measurement\em.  These are quantum processes:  
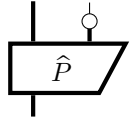
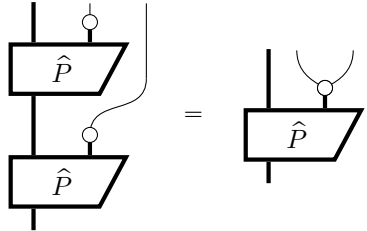
\ctikzfig{measdiag1}
satisfying:
\[
  \input{./figures/diagramsUUorth3.tikz}
\]
What this diagrammatic equation tells us is that performing the von Neumann measurement twice produces the same result as performing it once and then simply copying the outcome. That is, after performing a von Neumann measurement, any subsequent performances will yield the same outcome, and have no further effect on the quantum state. This defining property of von Neumann measurements is called \em von Neumann's projection postulate\em.

It's easy to see that measurements of the form (\ref{eq:non-dem-map}) are von Neumann measurements, for example, in the case of (\ref{eq:CQMII-7}) we can set:
  \ctikzfig{P-spider}
since, using bastard spider-fusion we have:    
  \ctikzfig{onb-is-vN}
Then, again using bastard spider-fusion:
  \ctikzfig{onb-vN-orth}
so we indeed obtain a von Neumann measurement.  

The focus on von Neumann measurements is historical in origin.   However, there isn't really any reason not to call any process on a quantum system that provides us with some classical data a quantum measurement.    So a demolition quantum measurement is really just a generic process from a quantum to a classical system: 
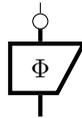
\ctikzfig{demo-povm}
Standard terminology for this is a \em demolition POVM-measurement\em. By purification we obtain a non-demolition counterpart:
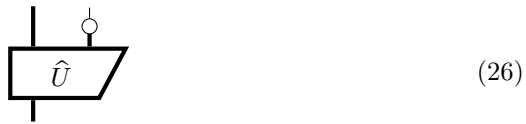
\begin{equation}\label{eq:povm-non-dem}
  \input{./figures/POVMg.tikz}
\end{equation}
where $\widehat U$ is an isometry (which follows immediately from causality, cf.~Section 5.2 in \cite{CKpaperI}).
Now, just by staring at this picture we obtain the following result:

\begin{theorem}[Naimark dilation]
Every non-demolition POVM-measurement arises as an isometry $\widehat U$ with a  non-degenerate measurement at one of its outputs.
\end{theorem}

Indeed this theorem just follows from two different readings of \eqref{eq:povm-non-dem}:
\ctikzfig{CQMII-11}


\subsection{Classical data flow}

In \cite{CKpaperI} we saw the following description of quantum teleportation:
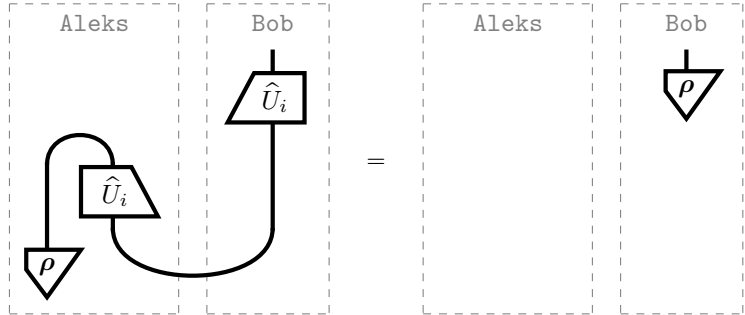
\ctikzfig{CQMII-12}
Here, the left-most index $i$ in the LHS refers to a measurement outcome, while the right-most one refers to the dependency of that operation on the measurement outcome.  We are now in a position to replace these indices with a classical wire. One tiny cosmetic change we will make is to slide the left-most box across the cup:
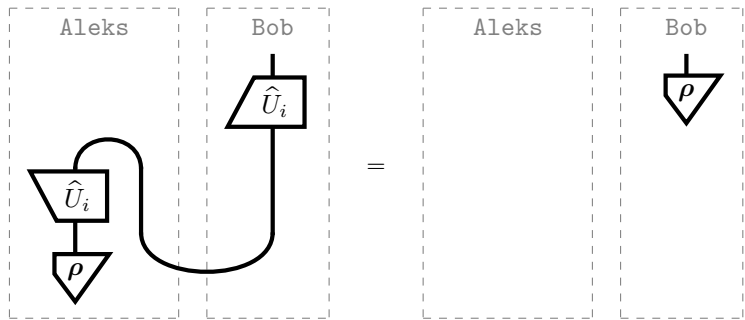
\ctikzfig{CQMII-13}
This will guarantee the classical measurement output is indeed an output.

In the previous section we already discussed what a measurement is, but what about the measurement-outcome dependent \em controlled \em operation?  In particular, what does it mean for a collection of unitaries to be controlled by a classical input?  One can capture this kind of process, called a \textit{controlled unitary} as a classical-quantum process:
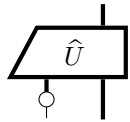
\ctikzfig{controliso1}
satisfying two equations:
\beq\label{eq:diagramsWW3copy}
\input{./figures/diagramsWW3copy.tikz}\qquad\qquad\quad\input{./figures/diagramsWW4copy.tikz}      
\eeq
In each of the equations, the classical input on the LHS is copied and used to pick the same unitary out of each occurence of $\widehat U$.  These are then composed in the two possible orders, in order to obtain the two defining equations (\ref{eq:unitary}) of unitarity.  The presence of deleting in each of the two RHS witnesses that for each pick of unitary we always get the same composite, namely a plain doubled wire.     

Taking the left-most of these two equations and inserting a uniform probability distribution (\ref{eq:uniform-probability-distribution}) in the classical input we obtain:
\ctikzfig{diagramsWW3copycopy}
Just a bit of un-yanking and some boxing now yields:  
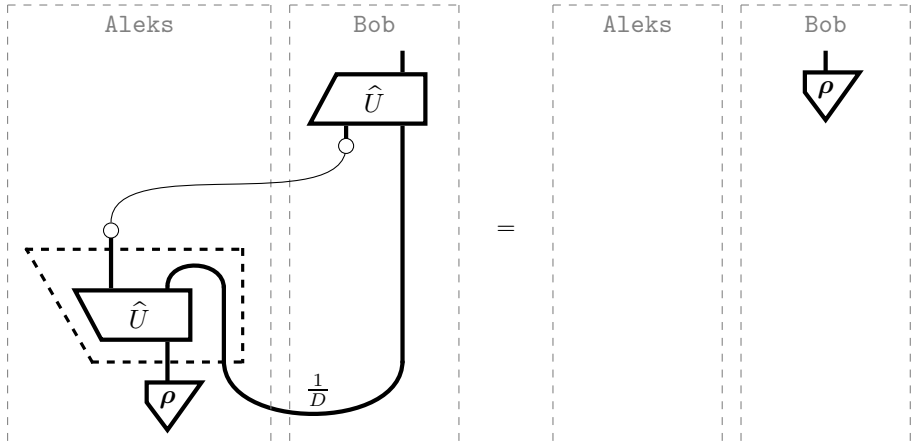
\ctikzfig{CQMII-14}
The only thing we haven't shown is that what is intended to be a measurement is indeed a measurement.  For it to be a non-degenerate measurement this would mean:
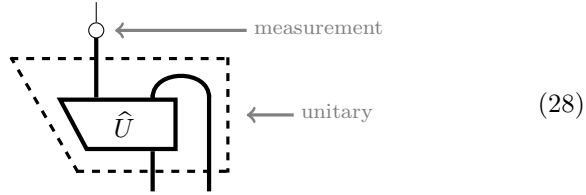
\beq\label{eq:cq-bell-meas-mapBIS}
\input{./figures/cq-bell-meas-mapBIS.tikz}
\eeq

Unitarity of course means that composing in either direction yields plain wire(s). Specialising this to the unitary in \eqref{eq:cq-bell-meas-mapBIS} yields two new equations:
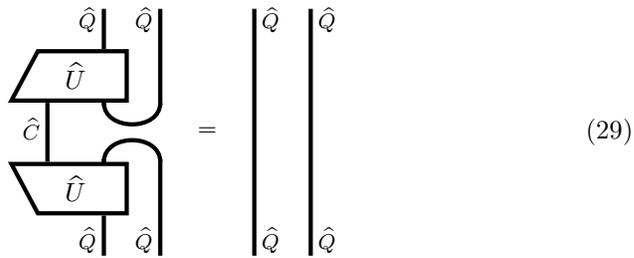
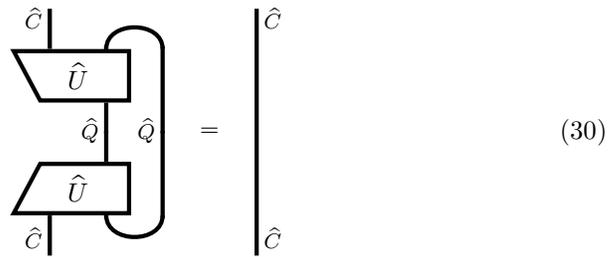
\beq\label{eq:CQMII-15}
\input{./figures/CQMII-15.tikz}
\eeq
\begin{equation}\label{eq:meas-unitary2}
  \input{./figures/CQMII-20.tikz}
\end{equation}
From this equation between diagrams it is immediately obvious that the dimension type $\widehat C$ should be large enough to realise two plain wires of type $\widehat Q$.  Concretely, this means that we need $D^2$-dimensional classical data for teleporting a $D$-dimensional quantum system.

We'll now demonstrate that the classical communication is crucial to performing teleportation, using the diagram above. Namely, let's see what happens if we don't use the measurement outcome, i.e.~delete it:    
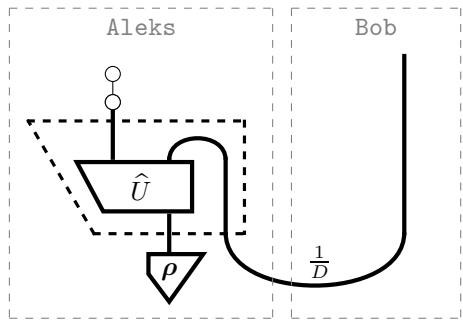
\ctikzfig{CQMII-17}
By causality we then have:
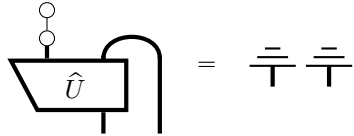
\ctikzfig{CQMII-16}
so as a result, all Bob obtains is maximal mixture:  
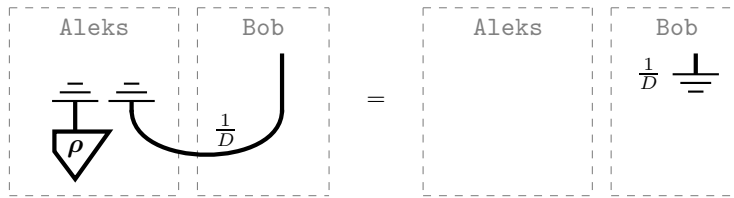
\ctikzfig{CQMII-18}


Turning the quantum inputs/outputs of equation~\eqref{eq:meas-unitary2} into classical inputs/outputs via encode/measure yields:
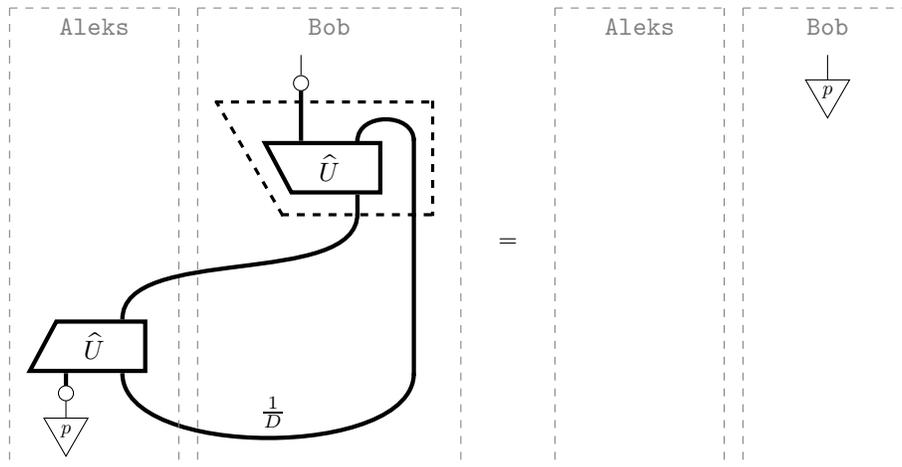
\ctikzfig{CQMII-19}
We have suggestively drawn in some boxes to make this look like a new protocol. But what does this protocol achieve? It transmits $D^2$-dimensional classical data using only a $D$-dimensional quantum channel, provided Aleks and Bob share a Bell state.  The is called \em (super)dense coding\em. 

One more variation on the same theme is the following:
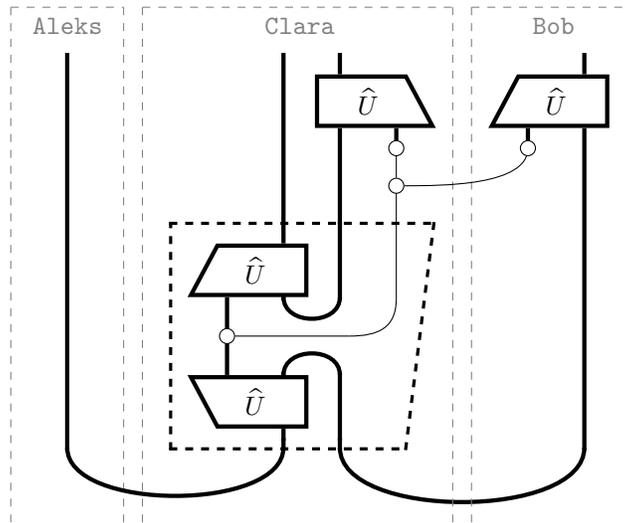
\ctikzfig{swapwireonly1CQM}
Using spider fusion and  the first equation of  (\ref{eq:diagramsWW3copy}) we obtain:    
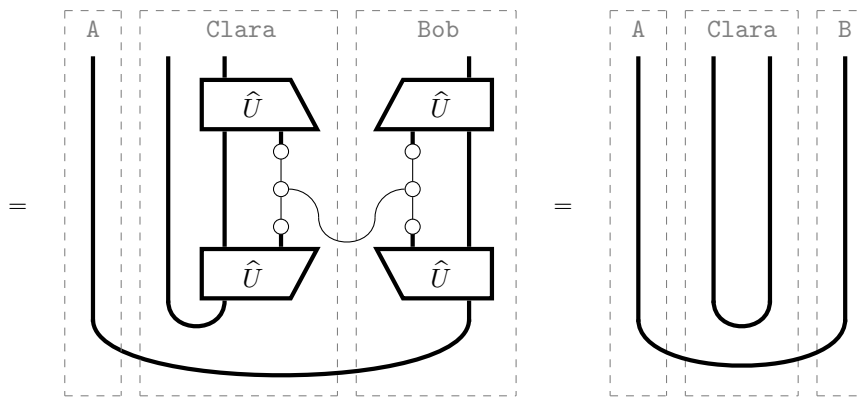
\ctikzfig{swapwireonly2CQM}
So what does this achieve?  We start with a pair of Bell-states:  
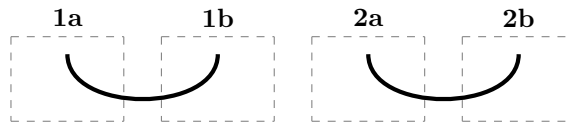
\[
\input{./figures/qwertyu.tikz}
\]
...and end up with a pair of Bell-states:  
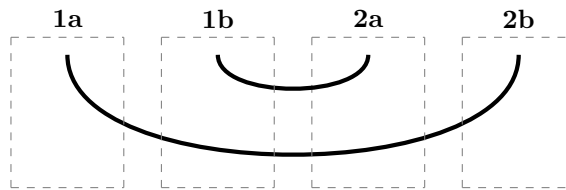
\[
\input{./figures/fishdyn5bis.tikz}
\]
However, while originally system {\bf 1a} was entangled with {\bf 1b} and {\bf 2a} with {\bf 2b}, the final state has {\bf 1a} entangled with {\bf 2b} and {\bf 1b} with {\bf 2a}. In other words, the entanglements have been `swapped'. This procedure is called  \em entanglement swapping\em.  The amazing bit about this is that {\bf 1a} becomes entangled to  {\bf 2b} while these systems were never acted upon together.  In other words, quantum theory allows for \em entangling without touching\em.

Note that while the teleportation protocol captures in essence the defining equation of string diagrams: 
\[
\input{./figures/line_yank3.tikz}
\]
which involves two caps/cups in LHS, entanglement swapping is the next one in line. It captures an equation involving three caps/cups in LHS:
\[
\input{./figures/line_yank2.tikz} 
\]

A practical use of entanglement swapping for  quantum technologies  is to generate 
entanglement over a large distance, given that one possesses some entangled states over shorter distances.  A device that performs entanglement swapping for this particular purpose is sometimes called a \textit{quantum repeater}, and is a crucial component to the feasibility of producing high-quality entangled states over long distances. It is called a repeater by analogy to classical signal repeaters, which make it possible to send messages long distances by occasionally capturing the signal and `repeating' an amplified version of the signal down the wire.

\subsection{Mixing}

Sometimes we don't have precise knowledge about which of a collection of quantum processes actually occurred. We can express a collection of quantum processes as a single quantum process controlled by a classical input. Then, our uncertainty can be represented by plugging a probability distribution into this input, which gives:

\begin{definition}
A \em mixture \em of quantum processes is a diagram of the form:
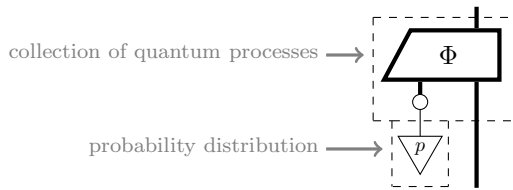
\begin{equation}\label{eq:mixture}
\input{./figures/CQMII-4.tikz}
\end{equation}
\end{definition}

Using the defining equation of probability distributions (\ref{eq:measure-rho}) and (bastard) spider-fusion,  mixtures can always be redrawn as follows: 
\beq\label{eq:CQMII-5}
\input{./figures/CQMII-5.tikz}
\eeq
That is, mixtures can be equivalently interpreted as discarding an output of a quantum system.

\begin{convention}\label{conv:full-support}
To simplify the proofs to follow, we will always assume mixtures have \textit{full support}. That is, for the probability distribution $p$ in \eqref{eq:mixture}, there is always another (possibly non-causal) classical state $p^{-1}$ such that:
\beq\label{eq:bekan12}
\input{./figures/bekan12.tikz}  
\eeq
Concretely, if $p$ consists of probabilities $p_i$, having full support means that all $p_i \neq 0$. Hence, we can form $p^{-1}$ whose entries are $1/p_i$. Assuming full support is no loss of generality, since we can always get rid of $0$'s by just choosing a smaller classical system for the mixture.
\end{convention}

Of particular interest are mixtures of pure quantum processes, and these can always be interpreted as discarding an output of a pure quantum process, simply by also `purifying' the state $\boldsymbol\rho$ in the RHS of (\ref{eq:CQMII-5}):
\ctikzfig{CQMII-6}
However, the converse is not true, that is, not every quantum process, which by (\ref{eq:quantummapx_CQMII}) all arise by discarding an output of a pure quantum process, can be interpreted as a mixture of pure quantum processes. In particular, discarding itself cannot be interpreted as a mixture of pure quantum effects, for the simple reason that there aren't any pure quantum effects (cf.~Theorem 5.2 in \cite{CKpaperI}).

We now use diagrammatic reasoning to prove something that is usually seen as a geometric property of quantum processes, namely that pure processes are `extremal' in the sense that no pure process can be decomposed as a non-trivial mixture of other processes.

\begin{proposition}\label{prop:regularmixing}
If a mixture is pure:
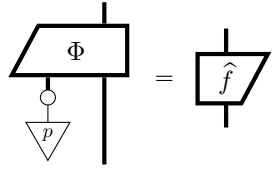
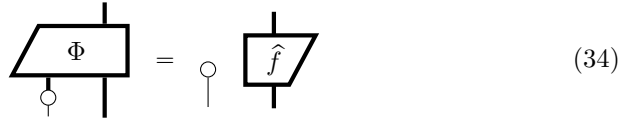
\[
\input{./figures/bekan13.tikz} \ \ = \ \dmap{\widehat f}  
\]
then: 
\beq\label{eq:regularmixing2}
\input{./figures/bekan15.tikz}\ \  = \ \   \input{./figures/bekan13a.tikz}
\eeq
\end{proposition}   
\begin{proof}
First, we represent any mixture as a reduced cq-map:
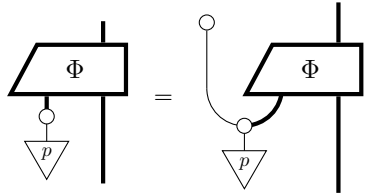
\[
\input{./figures/bekan13.tikz}\ \ =\ \ \input{./figures/bekan18class.tikz}
\]
By Proposition \ref{prop:discard-mix-pure-proc-class} the quantum process must separate as follows:  
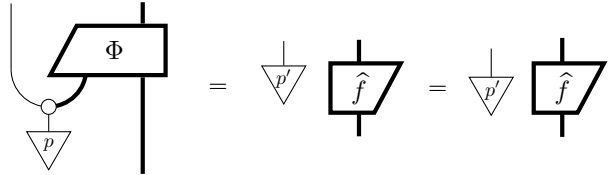
\[
\input{./figures/bekan20class.tikz}\ \ =\ \point{p'\!}\ \dmap{\widehat f}    
\]
We assume $p$ has full support, so apply $p^{-1}$ to both sides:
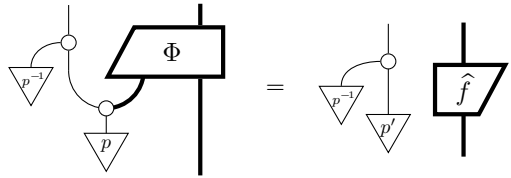
\ctikzfig{bekan24class-a}
Using spider-fusion  the LHS becomes:  
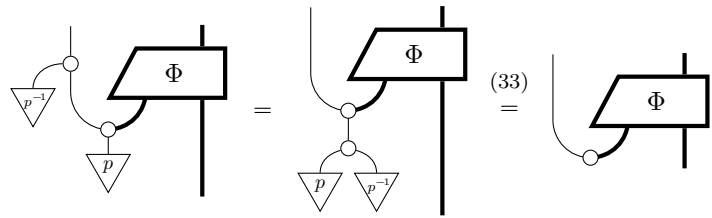
\[
\input{./figures/bekan24class.tikz}\ \ = \ \ \input{./figures/bekan25class.tikz}\ \ \namedeq{(\ref{eq:bekan12})} \ \ \input{./figures/bekan26.tikz}
\]
so bending the classical wire down gives:
\begin{equation}\label{eq:bacon-mix}
  \input{./figures/bacon1.tikz}
\end{equation}
Finally, the fact that the classical effect in \eqref{eq:bacon-mix} must be deleting follows from causality of $\Phi$ (which also implies causality of $\widehat f$).  This can be seen by plugging any causal state $\bm\rho$ into the quantum input and discarding the output:
\ctikzfig{bacon2}
\end{proof}

\subsection{Entanglement}  

For pure quantum states \em entanglement \em simply means non-separability.  That is, a bipartite state is \underline{not} entangled if it $\otimes$-separates as follows:
\[
\dbistate{\widehat\psi}\ =\ \dkpoint{\psi_1}\ \dkpoint{\psi_2}
\]
The tricky bit about entanglement for impure states, is that by means of  mixing we can obtain $\otimes$-non-separable processes even in cases where the part of the process connecting the two systems is entirely classical:
\begin{equation}\label{eq:bipmix}
  \input{./figures/bipmix.tikz}
\end{equation}
We say that these processes are \textit{classically correlated}. Such a connection is of a completely different nature from, for example, the quantum cups that we can exploit for deriving all kinds of quantum features.

Thus, to properly define entanglement, we should say not only that a quantum state doesn't separate, but also that it is not merely connected by classical correlations:

\begin{definition}\label{prop:entanglement}
A bipartite quantum state $\boldsymbol\rho$ is \em entangled \em if it \underline{cannot} be written in the following form for some $\Phi_1$ and $\Phi_2$:
\beq\label{eq:bipmix3}
\input{./figures/bipmix3.tikz}
\eeq
If a state is not entangled then we call it \em disentangled\em.  
\end{definition}

Even though we define disentangled states using the classical cup, we could equally well use any classical correlations, as we did in \eqref{eq:bipmix} above:  

\begin{proposition}\label{prop:entanglementsimple} 
A bipartite quantum state $\boldsymbol\rho$ is entangled if it \underline{cannot} be written in the following form for some $\Phi_1$,  $\Phi_2$ and probability distribution $p$:
\beq\label{eq:bipmix2}
\input{./figures/bipmix2.tikz}
\eeq
\end{proposition}
\begin{proof}
If a state is disentangled as in Definition \ref{prop:entanglement} then:
\ctikzfig{bipmix5}
so it matches the form (\ref{eq:bipmix2}). Conversely, we have:
\ctikzfig{bipmix4} 
so it is indeed in the form (\ref{eq:bipmix3}). 
\end{proof}

...and pure state entanglement is indeed a special case:

\begin{proposition}  
If a \underline{pure} bipartite quantum state is disentangled in the sense of Definition \ref{prop:entanglement}, then it is $\otimes$-separable.      
\end{proposition}
\begin{proof}
Represent any disentangled quantum state as follows: 
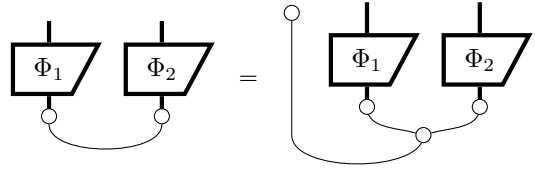
\[
\input{./figures/Ent2.tikz}\ \ = \ \  \input{./figures/Ent3.tikz} 
\]
Then, by Proposition \ref{prop:discard-mix-pure-proc-class} we have:
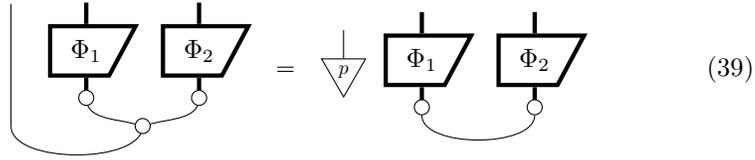
\beq\label{eq:Ent5}
 \input{./figures/Ent5.tikz} \ \ = \  \, \point{p}\,  \ \input{./figures/Ent2.tikz}    
\eeq
with $p$ a probability distribution.   For any deterministic classical state $i$ such that:
\[
\innerprodmap{i}{p}\   \neq \ 0
\] 
we then have:
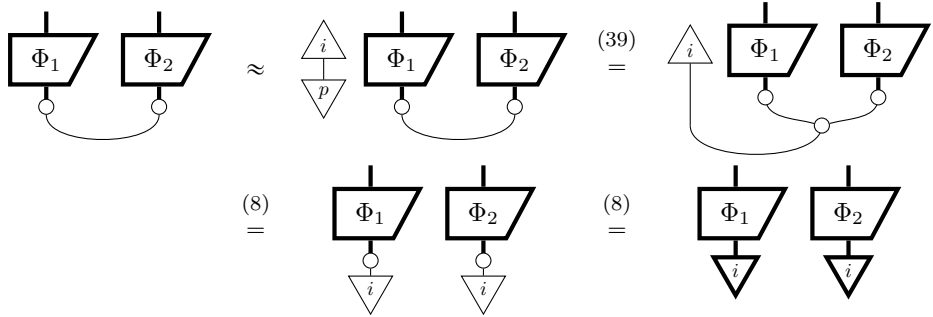
\[ \begin{array}{ccccc}
\input{./figures/Ent2.tikz} 
& \scalareq &
 \!  \innerprodmap{i}{p}\ \, \input{./figures/Ent2.tikz} 
& \namedeq{(\ref{eq:Ent5})} &
  \input{./figures/Ent6.tikz}  \\
& \namedeq{(\ref{eq:copyableonly})} &
   \input{./figures/Ent6a.tikz}
& \namedeq{(\ref{eq:copyableonly})} &
  \input{./figures/Ent7.tikz} 
\end{array} \]
\end{proof}

Since the quantum cup is pure and $\otimes$-non-separable, we obtain:

\begin{corollary}
The quantum cup is entangled.          
\end{corollary}
  
Having diagrammatic forms at hand for many quantum features, we can now easily investigate how these relate.  For example, what happens if we apply decoherence to one of two systems in a Bell-state?  It disentangles:
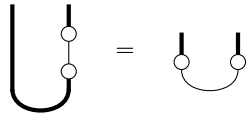
\ctikzfig{bipmix7}
This is an instance of a more general fact: 

\begin{theorem}\label{thm:deckillsentanglement}
Decoherence destroys entanglement, that is, if we apply decoherence to one of two systems in an entangled state, then it disentangles.
\end{theorem}
\begin{proof}
We have:
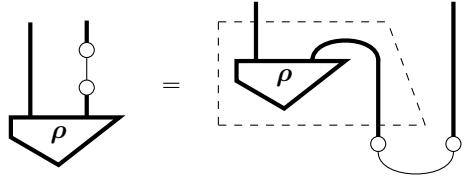
\ctikzfig{bipmix6}
\end{proof}

\subsection{Classifying entanglement}\label{sec:classify}  

Entanglement can be used as a resource for all sorts of tasks, ranging from secure communication to experimental demonstrations of quantum non-locality. But what sorts of entangled states are good for what tasks? To make this problem simpler, rather than considering all states, we can instead only think about states which are really distinct in terms of their capabilities. Thus, we should say what it means for entangled states to be equivalent, or more precisely: inter-convertible. In fact, there is not just one possibility, but many possibilities, depending on which operations we allow for converting one state into another. To borrow the terminology of resource theories, we call these \textit{free processes}, implying that they are sufficiently easy to realise that they do not add any `value' to the resource we already have. In the case of entanglement, free processes should, at the very least, not introduce any new entanglement.  

To identify what these processes should look like, let us suppose that some number of parties only have access to a single system of a quantum state. For instance, with two parties, the  situation looks like this:
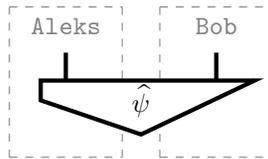
\ctikzfig{resources-state}
Then, the most general family of (causal) processes which cannot introduce any new entanglement consist of \textit{\underline{L}ocal \underline{O}perations}, i.e.~quantum processes localised just to a single party:      
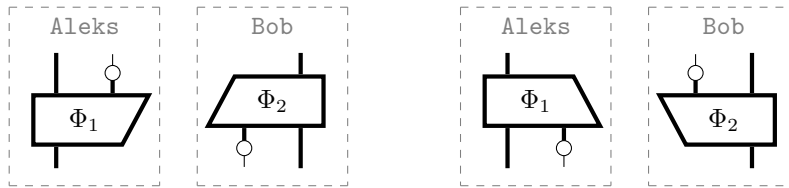
\ctikzfig{resources18}
and \textit{\underline{C}lassical \underline{C}ommunication}, i.e.~classical wires connecting the two parties:
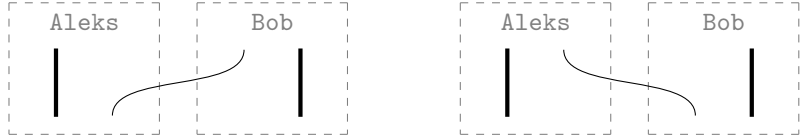
\ctikzfig{snake7}
We refer to any composition of such processes as a \textit{LOCC process}, if a state $\widehat\psi$ can be converted into a state $\widehat\psi'$ by means of a LOCC process, it is called \textit{LOCC-convertible}, and mutually LOCC-convertible processes are called \textit{LOCC-equivalent}.  For bipartite pure states, LOCC-convertibility is well-understood, and can be totally characterised by the \textit{majorization pre-ordering} (see further reading below), but this becomes much more difficult for three or more systems. 

Fortunately, one can coarse-grain LOCC by dropping the restriction that the local processes above should be causal. A consequence of this is that we can now rely on local processes which may only occur with some non-zero probability, therefore one now speaks of \textit{\underline{S}tochastic local operations and classical communication} (SLOCC). In fact, once we drop causality, classical communication becomes redundant:

\begin{theorem}\label{thm:slocc-no-cc}
A bipartite state $\widehat\psi$ can be SLOCC-converted into a bipartite state $\widehat\psi'$ if and only if there exist (not necessarily causal) quantum processes $\Phi_1$ and $\Phi_2$ such that:
  \begin{equation}\label{eq:slocc-sep}
   \input{./figures/slocc-sep.tikz}
  \end{equation}
\end{theorem}
\begin{proof}
Note that, since we are not requiring quantum processes to be causal, we are allowed to use classical caps, so the direction of classical communication is irrelevant:
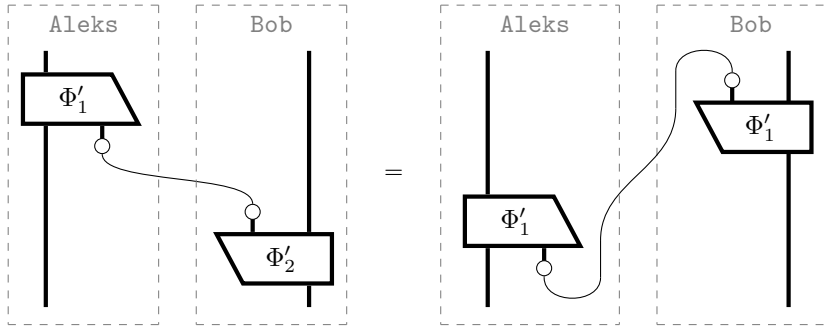
\ctikzfig{slocc-flip}
Thus, if $\widehat\psi$ is SLOCC-convertible into $\widehat\psi'$, we can assume without loss of generality that there exist cq-maps $\Phi_1'$ and $\Phi_2'$ such that:
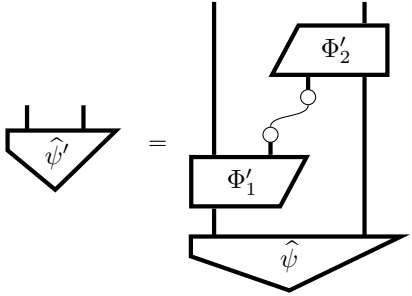
\[
\dbistate{\widehat\psi'\!}  \ \ =  \,  \ \input{./figures/slocc-branch2-copy1.tikz}
\]
for some sufficiently large classical system. Then, since:
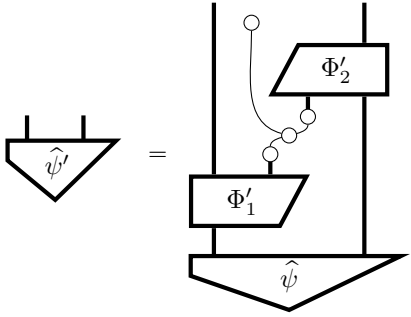
\[
\dbistate{\widehat\psi'\!}  \ \ =  \,  \ \input{./figures/slocc-branch2-copy2.tikz}
\]
by Proposition \ref{prop:discard-mix-pure-proc-class} we have:  
\beq\label{slocc-branch3}
  \input{./figures/slocc-branch3.tikz}  
\eeq  
for some $p$.  Since $p$ is a (causal) probability distribution, there must be some ONB effect $i$ such that:
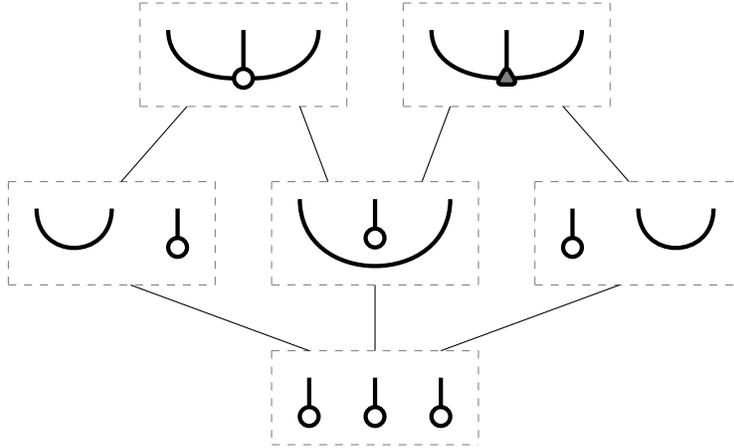
\beq\label{bekan777}
  \pointbraket{p}{i}\ \scalareq\ \emptydiag 
\eeq
Hence:
\[
\dbistate{\widehat\psi'\!}  \ \ \namedscalareq{(\ref{bekan777})}  \,  \   \input{./figures/slocc-branch4.tikz}
  \]
Letting:
\[
\dmap{\Phi_1}\ \, \scalareq\ \ \input{./figures/slocc-L.tikz}
\qquad\qquad
\dmap{\Phi_2}\  \, \scalareq\ \ \input{./figures/slocc-R.tikz}
\]
(up to an appropriately chosen number) yields equation \eqref{eq:slocc-sep}.  
\end{proof}

The passage from LOCC to SLOCC drastically cuts down the number of equivalence classes of states. For example, with pure bipartite qubit states, there are only two:
\ctikzfig{bekan69}
where a line downward indicates that one class of states is SLOCC-convertible into another. For example, we can obtain the separable state pictured above from the cup via:
\ctikzfig{slocc-cup}
In particular, there exists a single `best', i.e.~\textit{SLOCC-maximal}, state from which we can obtain any other bipartite qubit state via SLOCC operations. 

However, for 3 qubits, the situation becomes much more interesting. Now, there is no longer a single SLOCC-maximal state, but two:
\ctikzfig{slocc3}  
These represent two \underline{qualitatively different} kinds of tripartite entanglement, which can be witnessed by distinct behaviours of the different kinds of `arachnids' from which they arise. The first kind of SLOCC-maximal state is the GHZ-state, which as we saw in Section~\ref{sec:quantum-and-bastard}, arises from a (doubled) 3-legged spider. The second kind of state, the \textit{W-state} arises from a different kind of arachnid, called an \textit{anti-spider}:

\begin{definition}
  Anti-spiders are families of processes:
  \[ 
  \left\{\  \underbrace{\overbrace{\ \input{./figures/spiderW.tikz}\ }^n}_m \ \right\}_{m,n} 
  \]
  satisfying three spider-like equations:
  \[ 
  \input{./figures/spider-flipW.tikz}\ \ =\ \ \input{./figures/spiderW.tikz}\ \ =\ \ \input{./figures/spider-flipW2.tikz} 
  \]
  \ctikzfig{bastards-compose-bisW1}
  as well as one (very un-spider-like) equation:
  \begin{equation}\label{eq:explode2}
  \input{./figures/bekan84.tikz}\ \ \scalareq\ \ \input{./figures/bastards-compose-trisW.tikz}    
  \end{equation}
\end{definition}

In fact, these two very different kinds of processes, spiders and anti-spiders, precisely capture the distinction of the GHZ-state vs.~the W-state:

\begin{theorem}\label{thm:ghz-w-slocc}
  For any family \whitedot of spiders and \wdot of anti-spiders on qubits, the following states:
  \[
 \input{./figures/qspiderb.tikz} \qquad\qquad\qquad\qquad \input{./figures/dW_copy.tikz} 
  \]
  are always SLOCC-equivalent to the GHZ- and W-state, respectively.
\end{theorem}

Conversely, since the GHZ-state and the W-state are the \textit{only} possible non-separable tripartite states, up to SLOCC-equivalence, it is straightforward to show that any family \blackdot of processes satisfying:
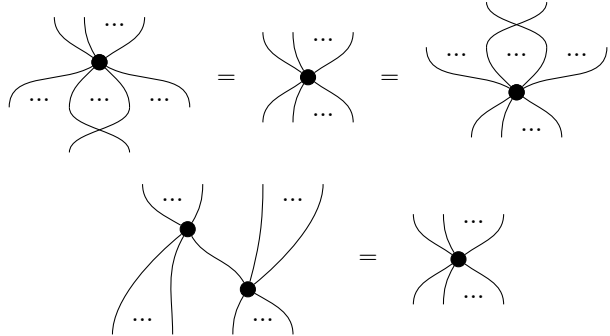
\[
\input{./figures/spider-flip-gen.tikz}\ \ =\ \ \input{./figures/spider-gen.tikz}\ \ =\ \ \input{./figures/spider-flip-gen2.tikz}
\]
\[
\input{./figures/bastards-compose-bis-gen.tikz}
\]
must either be isomorphic to some fixed spider \whitedot or anti-spider \wdot. That is, there must exist some isomorphism $f$ such that either:
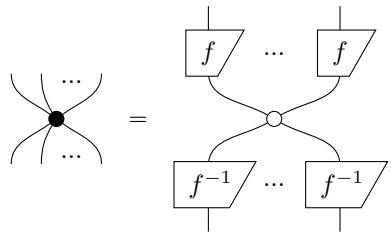
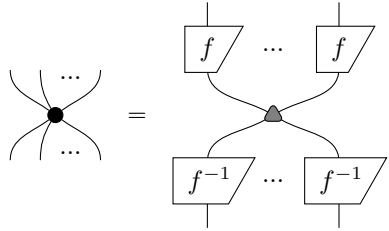
\[ 
\input{./figures/spider-gen.tikz}\ \ =\ \ \input{./figures/arachnid-ghz.tikz} 
\]
or
\[ 
\input{./figures/spider-gen.tikz}\ \ =\ \ \input{./figures/arachnid-w.tikz} 
\]
We refer the interested reader to further reading below for proofs of this fact and Theorem~\ref{thm:ghz-w-slocc}.

\subsection{Category-theoretic counterpart}\label{sec:comodule} 

One can show that $\widehat P$ defines a von Neumann measurement if and only if the (non-doubled) process $P$ can be chosen in such a way that it forms a \textit{dagger-comodule} with respect to \whitedot. That is, $P$ satisfies:
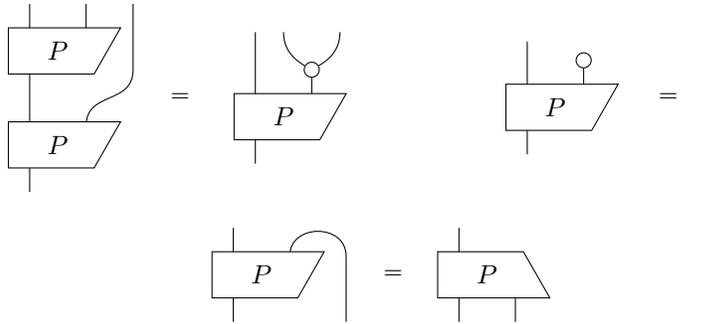
\[
\input{./figures/comodule-assoc-cqm.tikz}
\qquad\qquad\qquad
\input{./figures/comodule-unit-cqm.tikz}
\]
\[
\input{./figures/comodule-dagger.tikz}
\]
if and only if the associated quantum process satisfies:
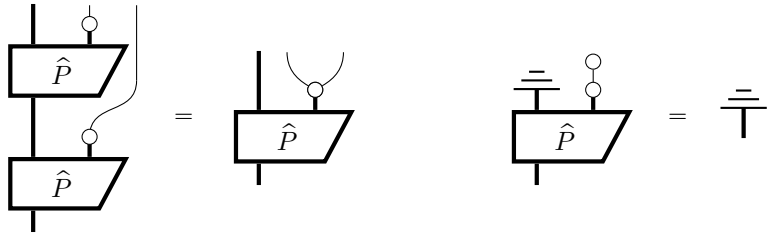
\[
  \input{./figures/diagramsUUorth3-cqm.tikz}
  \qquad\quad\qquad\quad
  \input{./figures/measdiag1-caus-cqm.tikz}
\]
Going from dagger-comodules to von Neumann measurements is straightforward, whereas coming back is a bit more subtle, and the requires the correct choice of $P$ in:  
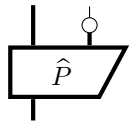
\ctikzfig{measdiag1}
which is only uniquely fixed up to a phase gate (introduced in Section \ref{sec:phase-group} below).   

Just as spiders correspond to special commutative Frobenius alegbras, anti-spiders correspond to \textit{anti-special commutative Frobenius algebras}. The crucial difference comes from what happens when two legs meet, or in algebraic language, when a multiplication meets a comultiplication:
\[ 
\input{./figures/black-copymatch.tikz}\ \ =\ \ \textrm{???} 
\]
For special commutative Frobenius algebras, this yields an identity, whereas for anti-special Frobenius algebras, this separates:
\[ 
\input{./figures/copymatch.tikz} \quad\qquad\qquad\textrm{vs.}\qquad\qquad\quad \input{./figures/as-copymatch.tikz}   
\]

\subsection{Reference and further reading}

von Neumann measurements, not surprisingly, are due to von Neumann \cite{vN}, in particular including the idea of the collapse.  A diagrammatic treatment first appeared in  \cite{CPav, CPaqPav}, including the comodule perspective of Section \ref{sec:comodule}.  The treatment of von Neumann measurements as quantum processes appears in \cite{CKbook}, and until recently, the equivalence of the two perspectives was an open problem. This has now been settled, and is in the process of being written up \cite{CGKvN}. Naimark's theorem for POVMs first appeared in \cite{Naimark}.  

The quantum teleportation protocol first appeared in \cite{tele}, dense coding was first proposed in \cite{BW}, and entanglement swapping in \cite{Swap}.  Their  diagrammatic treatment using `classical wires' was initiated in \cite{CPav, CPaqPav}.

Nielsen \cite{Nielsen} gave a characterisation for LOCC-convertibility of bipartite states in terms of the majorization pre-ordering briefly mentioned in Section~\ref{sec:classify}. The SLOCC-classification of three qubits is taken from \cite{DurVC}.  From four qubits onward there is an uncountably infinite set of SLOCC classes.  One can nevertheless still identify finitely many parametrised `super-classes' (see e.g.~\cite{Verstraete4qubit, lamata2007inductive}). 

We defined anti-spiders (a.k.a.~anti-special commutative Frobenius algebras) in \cite{CK}, and proved the   correspondence of special and anti-special commutative Frobenius algebras to GHZ- and W-states in Theorem~\ref{thm:ghz-w-slocc}.

\section{Quantumness}\label{sec:quantumness}  

\begin{quote}
\em When you think outside the bun.  \\
It's scary out there.\\
There's also lots of stinky poo.  \\
But sometimes when you go outside of your comfort zone. \\     
You might find something brand new.\\
To help the humans leap.\\
Forward, do a simple leap.\\
Never be they same, they leap.\\
Across the ****ing annals of time.\\
Quantum leap, leap.\par \em \hfill    --- Tenacious D, lyrics of ``Quantum leap", 2012.                  
\end{quote}

\noindent  
We extensively discussed what classicality is, and in particular, that it is witnessed by spiders.  But what is non-classicality, i.e.~\em genuine quantumness\em?  This section is about the stuff that has no classical counterpart.  There will be no need to introduce any new diagrammatic ingredient.  All we need to define quantumness genuine quantumness are the same spiders that we use to witness classicality.

\subsection{Unbiasedness and phase states}

What does it mean for something to be genuinely quantum? Our answer is:  when in cannot survive the quantum-classical passage.  Since  quantum-classical passage means measurement, and since we defined measurement in terms of spiders, we will again rely on spiders. We will denote a family of spiders  by \whitedot.    

\begin{definition}
A \em phase state \em for \whitedot is a pure  state $\widehat\psi$ that satisfies:  
\beq\label{eq:unbiased6}
\input{./figures/unbiasedstate6.tikz}
\eeq
\end{definition}

In the LHS of equation (\ref{eq:unbiased6}) we see the state $\widehat\psi$ being converted into a classical system by means of a measurement, and in the RHS we see that nothing remains of it.  We provide a new notation for phase states, namely:  
\[
\dphasepoint{\alpha} 
\]
and its transpose, conjugate and adjoint are denoted as follows as:
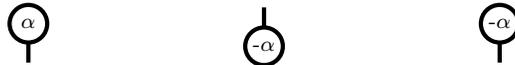
\[
\dphasecopoint{\alpha}\qquad\qquad\qquad\quad\dphasepoint{\!\textrm{-}\alpha\!}\qquad\qquad\quad\qquad\dphasecopoint{\!\textrm{-}\alpha\!}
\]
Unlike all the states we worked with thus far, phase states are not causal:
\ctikzfig{CQMII-21}
However, fixing this is simply a manner of adjoining a number:
\[
\oneoverD\ \dphasepoint{\alpha} 
\]
The reason for making this convention will become clear in the following section.

Phase states have an interesting property, which becomes exposed by plugging deterministic classical states in the measurement-outcome wire:
\ctikzfig{CQMII-22}
The number in the RHS does not depend on $\alpha$ at all, and is the same for all $i$.  This means that phase states have no `bias' towards any of the measurement outcomes.  therefore they are also called \em unbiassed\em. 

\subsection{Phase spiders}\label{sec:phase-spiders}

We can  use phase states to decorate quantum spiders, yielding \em phase spiders\em:
  \[
  \input{./figures/CQMII-27.tikz}\ \ :=\ \ \raisebox{0.7mm}{\input{./figures/CQMII-26.tikz}} 
  \]
and we call the decorations \em phases\em.  Note that it doesn't matter if we plug a phase state into an input or the transpose of a phase state into an output:
\ctikzfig{CQMII-25}
or if we plug it into a different leg. Just like in the case of phase states, if a phase spider in brought into contact with classicality, its phase vanishes:

\begin{proposition}\label{thm:alphaspider_intro}    
If any leg of a phase spider is measured,
its  phase vanishes:
\[ 
\input{./figures/CQMII-24.tikz}  
\] 
\end{proposition}
\begin{proof}
Using bastard spider-fusion we have:  
\ctikzfig{CQMII-23}
\end{proof}

...and the same goes for decoherence:
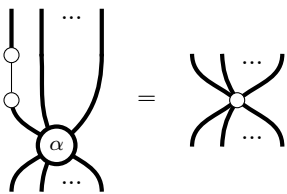
\[ 
\input{./figures/CQMII-28.tikz}
\] 
...or any other bastard spider:
    \ctikzfig{CQMII-29}
What about when two phase spiders fuse together? Using quantum spider fusion we have:  
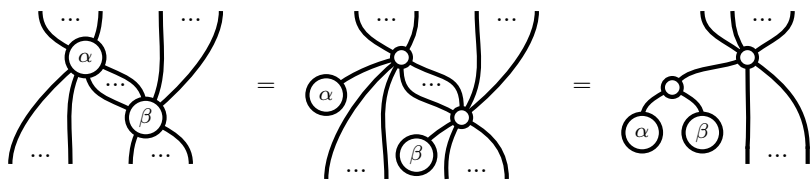
\ctikzfig{CQMII-30}
The RHS is again a phase spider:

\begin{lemma}\label{lem:mult-phase1}
The state:
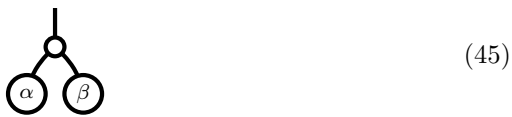
\begin{equation}\label{eq:phase-sum-pre}  
  \input{./figures/CQMII-31.tikz}
\end{equation}
 is  a phase state, and hence:  
 \ctikzfig{CQMII-32}
 is  a phase spider.
\end{lemma}
\begin{proof}
Using bastard spider fusion we have:
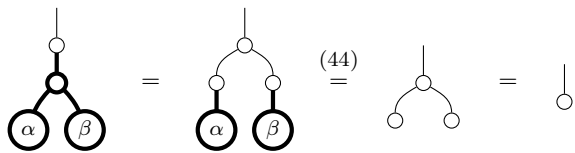
\ctikzfig{CQMII-33}
So equation (\ref{eq:unbiased6}) is indeed satisfied for the state~\eqref{eq:phase-sum-pre}.      
\end{proof}

Hence, phase spiders fuse as follows:   
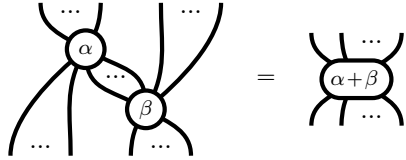
\[ 
\input{./figures/CQMII-34.tikz}
\] 
where we used the shorthand:
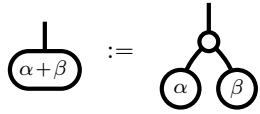
\[
   \begin{tikzpicture}
	\begin{pgfonlayer}{nodelayer}
		\node [style=white phase ddot] (0) at (0, -0.5) {$\ \alpha\! + \!\beta\ \,$};
		\node [style=none] (1) at (0, 0.75) {};
	\end{pgfonlayer}
	\begin{pgfonlayer}{edgelayer}
		\draw [style=boldedge] (0) to (1.center); 
	\end{pgfonlayer}
\end{tikzpicture}\ \ := \ \ \input{./figures/CQMII-36.tikz}  
\]
Extending this notation  to $n$ phases:
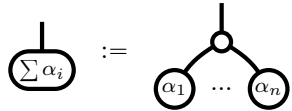
\[
   \begin{tikzpicture}
	\begin{pgfonlayer}{nodelayer}
		\node [style=white phase ddot] (0) at (0, -0.5) {$\ \sum\alpha_i\ $};
		\node [style=none] (1) at (0, 0.75) {};
	\end{pgfonlayer}
	\begin{pgfonlayer}{edgelayer}
		\draw [style=boldedge] (0) to (1.center); 
	\end{pgfonlayer}
\end{tikzpicture}\ \ := \ \ \input{./figures/CQMII-40.tikz}
\]
 any diagram consisting only of phase spiders and which is moreover connected is itself a phase spider,  whose phase is the sum of the phases of each of the component spiders:
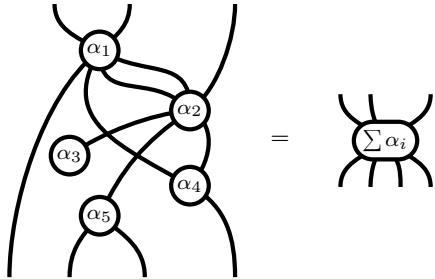
\[ 
\input{./figures/CQMII-38.tikz}
\] 

So genuine quantumness incarnates as decorations for quantum spiders.  Everything we already knew about classical spiders  carries over to these genuinely quantum spiders, with the added bonus that decorations are added when fusing.  In fact, there is something more to these decorations: they form a structure which mathematicians completely understand and physicists love.

\subsection{The phase group}\label{sec:phase-group}

Indeed, in this section, we will see that the set of phase states forms a \textit{commutative group}.  Taking phase states to be group elements:
\[ 
\dphasepoint{\alpha}\ \ \leftrightarrow \ \  a 
\]
we have already identified a candidate for the group-sum of two phase states:
\[
  \input{./figures/CQMII-36.tikz} \ \ \leftrightarrow \ \  a + b
\]
so the group-sum itself is:
\[
\input{./figures/sumdiag.tikz} \ \ \leftrightarrow \ \  +
\]
Hence,  it only remains to find the unit and the inverse. 

\begin{theorem}\label{thm:phase-group}
  For any family of spiders $\whitedot$, the set of phase states:    
  \[ 
\ \,  \left\{ \dphasepoint{\alpha} \right\}_{\alpha}
  \]
  form a commutative group where:  
  \bit
    \item the  group-sum is: 
 \[ 
     \begin{tikzpicture}
	\begin{pgfonlayer}{nodelayer}
		\node [style=white phase ddot] (0) at (0, -0.5) {$\ \alpha\! + \!\beta\ \,$};
		\node [style=none] (1) at (0, 0.75) {};
	\end{pgfonlayer}
	\begin{pgfonlayer}{edgelayer}
		\draw [style=boldedge] (0) to (1.center); 
	\end{pgfonlayer}
\end{tikzpicture}\ \ := \ \ \input{./figures/CQMII-36.tikz}\ \ 
\] 
   \item the  unit is:
   \beq\label{eq:phasGthm2}
    \begin{tikzpicture}
	\begin{pgfonlayer}{nodelayer}
		\node [style=none] (0) at (0, 0.5) {};
		\node [style=white ddot] (1) at (0, -0.5) {};
	\end{pgfonlayer}
	\begin{pgfonlayer}{edgelayer}
		\draw [style=boldedge] (1) to (0.center); 
	\end{pgfonlayer}
\end{tikzpicture}
   \eeq
   \item the  inverse is:
   \[ 
       \dphasepoint{\!\textrm{-}\alpha\!}
     \] 
    \eit
\end{theorem}  
\begin{proof}
Using bastard spider fusion we have:         
 \ctikzfig{phase-mult-pf-bis}
so (\ref{eq:phasGthm2}) is a phase state, and by conjugating both sides of (\ref{eq:unbiased6}) we obtain:  
\ctikzfig{CQMII-43}
so the conjugate of a phase state is again a phase state.  The inverse law arises by first unfolding~\eqref{eq:unbiased6}:
\ctikzfig{CQMII-44} 
and then doubling:
\[ 
  \input{./figures/CQMII-42.tikz}
\] 
It only remains to verify associativity, commutativity, and unit laws, which all follow from spider fusion, for example, in the case of associativity we have:        
  \ctikzfig{CQMII-41}
\end{proof}

We formulated the \em phase group \em in terms of phase states, but in fact, the decorations of any phase spider are subject to this group structure.  For example, one can reformulate Theorem \ref{thm:phase-group} in terms of \em phase gates\em:
\ctikzfig{CQMII-45}
These phase gates inherit sums from phase states:
\ctikzfig{CQMII-48}
as well as inverses:
\ctikzfig{CQMII-46}
Hence:

\begin{corollary}
For any family of spiders $\whitedot$, the set of phase gates:    
  \[ 
\ \,  \left\{  \begin{tikzpicture}
	\begin{pgfonlayer}{nodelayer}
		\node [style=white phase ddot] (0) at (0, 0) {$\alpha$};
		\node [style=none] (1) at (0, -1) {};
		\node [style=none] (2) at (0, 1) {};
	\end{pgfonlayer}
	\begin{pgfonlayer}{edgelayer} 
		\draw [style=boldedge] (2.center) to (0);
		\draw [style=boldedge] (0) to (1.center);  
	\end{pgfonlayer}
\end{tikzpicture} \right\}_{\alpha}  
  \]
forms a commutative group where:
   \bit
    \item the group-sum is: 
  \[
  \input{./figures/CQMII-50.tikz}\ \ \ \ \,
  \]
   \item the  unit is:
  \[
  \begin{tikzpicture}
	\begin{pgfonlayer}{nodelayer}
		\node [style=none] (0) at (-1.75, 1) {};
		\node [style=none] (1) at (-1.75, -1) {};
	\end{pgfonlayer}
	\begin{pgfonlayer}{edgelayer}
		\draw [style=boldedge] (0.center) to (1.center);
	\end{pgfonlayer}
\end{tikzpicture}\ \ \ \ \,
 \] 
   \item the  inverse is:
  
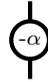
\[
  \begin{tikzpicture}
	\begin{pgfonlayer}{nodelayer}
		\node [style=white phase ddot] (0) at (0, 0) {$\textrm{-}\alpha$};
		\node [style=none] (1) at (0, -1) {};
		\node [style=none] (2) at (0, 1) {};
	\end{pgfonlayer}
	\begin{pgfonlayer}{edgelayer}
		\draw [style=boldedge] (2.center) to (0); 
		\draw [style=boldedge] (0) to (1.center); 
	\end{pgfonlayer}
\end{tikzpicture} \ \ \ \ \, 
    \]
      \eit
\end{corollary}

\subsection{A hint of non-locality}\label{sec:nonloc-hint}

When we apply a phase gate to each of three systems in a GHZ-state, by phase spider fusion we obtain:
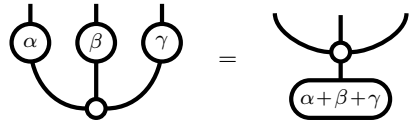
\ctikzfig{qspiderbGHZ} 
While this is a seemingly innocent looking diagrammatic equation,  when we interpret it physically the implications are somewhat shocking! Assume that the three parties that perform the three phase gates:   
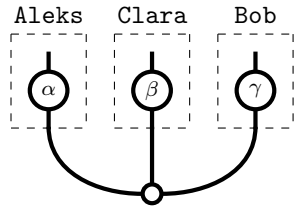
\ctikzfig{qspiderbGHZbisCQM}
are so far apart that when they perform their phase gates, light does not have the time to travel between them.   While the choices of the angles $\alpha$, $\beta$ and $\gamma$ are made independently, at very distant locations, the resulting state only depends on the group-sum of the three phases. So if, for instance, these phases would have been permuted, the resulting state would be the same. This hints at the fact that these processes are interacting instantaneously over a long distance. The diagram literally shows that it is \underline{as if} the three phases are travelling backward in time to meet up with each other:    
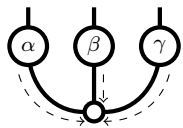
\ctikzfig{qspiderbGHZtris} 
Of course, all of this happens at the level of a quantum state, and naive measurement kills the phases, and hence all the magic:
 \ctikzfig{qspiderbGHZquadCQM}
However, in the third part of this series we will see that one can make clever measurement choices that actually expose this dance of the phases in their measurement outcomes, which gives rise to \em quantum non-locality\em.

\subsection{Reference and further reading}       

The notions of phase spiders and the phase group are taken from \cite{CD1,CD2}. In quantum information theory, unbiased states are typically studied in the context of a pair of ONBs where each element of the first ONB is unbiased with respect to the second ONB (or vice-versa, which is equivalent). These \textit{mutually unbiased bases} were first introduced by \cite{Schwinger}. An extensive survey of what is known about them is in \cite{MUBsurvey}, including problems concerning their classification.

\section{What comes next}    

By means of spiders (a.k.a.~generalised wires) we were able to capture classicality, as well as genuine quantumness.  In particular, in Section~\ref{sec:phase-spiders} we encountered phase spiders.  These will allow us to `fill in'  boxes of the kind:
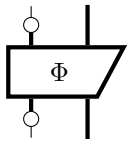
\ctikzfig{classquantmap1}
In \cite{CQMIII} we will consider not just one family of phase spiders, but two, yielding much more interesting ways to `fill in boxes':  
\ctikzfig{ZX-diag-form-scopy}
These are either related by \textit{complementary} or by \textit{strong complementary}, resulting in a spider-language that is \textit{universal} for describing quantum and classical processes. For strongly complementary pairs of phase spiders we will moreover  obtain a strong new completeness theorem for a large family of diagrams. Strongly complementary spiders will also enable us to complete the proof of quantum non-locality we hinted at in Section~\ref{sec:nonloc-hint},  and express many other things, including quantum communication and quantum computation.

\bibliography{main}
\bibliographystyle{plain}  

\end{document}